\documentclass[aps,twocolumn,pra,footnoteinbib,superscriptaddress]{revtex4-2}
\usepackage{graphicx}
\usepackage{epstopdf}
\usepackage{amsmath}
\usepackage{amsthm}
\usepackage{amssymb}
\usepackage{bm}
\usepackage{dsfont}
\usepackage{mathtools}
\usepackage{hyperref}
\usepackage{braket}
\usepackage{color}
\usepackage{multirow}
\usepackage{mathtools}

\def\>{\rangle}
\def\<{\langle}

\newcommand\numberthis{\addtocounter{equation}{1}\tag{\theequation}}

\DeclareMathOperator{\Tr}{Tr}
\DeclareMathOperator{\Id}{\mathds{1}}

\DeclareMathOperator*{\lprod}{\overleftarrow{\prod}}

\begin{document}

\title{Hardware-efficient ansatz without barren plateaus in any depth}

\author{Chae-Yeun Park}
\affiliation{Xanadu, Toronto, ON, M5G 2C8, Canada}

\author{Minhyeok Kang}
\affiliation{Department of Chemistry, Sungkyunkwan University, Suwon 16419, Korea}
\affiliation{SKKU Advanced Institute of Nanotechnology (SAINT), Sungkyunkwan University, Suwon 16419, Korea}
\affiliation{Institute of Quantum Biophysics, Sungkyunkwan University, Suwon 16419, Korea}

\author{Joonsuk Huh}
\affiliation{Xanadu, Toronto, ON, M5G 2C8, Canada}
\affiliation{Department of Chemistry, Sungkyunkwan University, Suwon 16419, Korea}
\affiliation{SKKU Advanced Institute of Nanotechnology (SAINT), Sungkyunkwan University, Suwon 16419, Korea}
\affiliation{Institute of Quantum Biophysics, Sungkyunkwan University, Suwon 16419, Korea}

\date{\today}

\newtheorem{proposition}{Proposition}
\newtheorem{example}{Example}
\newtheorem{lemma}{Lemma}
\newtheorem{theorem}{Theorem}
\newtheorem{observation}{Observation}
\newtheorem{corollary}{Corollary}
\newtheorem{conjecture}{Conjecture}
\newtheorem*{remark}{Remark}

\begin{abstract}
Variational quantum circuits have recently gained much interest due to their relevance in real-world applications, such as combinatorial optimizations, quantum simulations, and modeling a probability distribution.
Despite their huge potential, the practical usefulness of those circuits beyond tens of qubits is largely questioned.
One of the major problems is the so-called \textit{barren plateaus} phenomenon.
Quantum circuits with a random structure often have a flat cost-function landscape and thus cannot be trained efficiently.
In this paper, we propose two novel parameter conditions in which the hardware-efficient ansatz (HEA) is free from barren plateaus for arbitrary circuit depths.
In the first condition, the HEA approximates to a time-evolution operator generated by a local Hamiltonian. 
Utilizing a recent result by [Park and Killoran, Quantum \textbf{8}, 1239 (2024)], we prove a constant lower bound of gradient magnitudes in any depth both for local and global observables.
On the other hand, the HEA is within the many-body localized (MBL) phase in the second parameter condition.
We argue that the HEA in this phase has a large gradient component for a local observable using a phenomenological model for the MBL system.
By initializing the parameters of the HEA using these conditions, we show that our findings offer better overall performance in solving many-body Hamiltonians.
Our results indicate that barren plateaus are not an issue when initial parameters are smartly chosen, and other factors, such as local minima or the expressivity of the circuit, are more crucial.
\end{abstract}

\maketitle

By combining huge neural networks and parameter optimization techniques, machine learning has achieved great successes in diverse tasks such as image classification~\cite{krizhevsky2017imagenet}, defeating human level in playing games~\cite{silver2016mastering}, natural language processing~\cite{brown2020language}, and predicting protein structures~\cite{jumper2021highly}. 
Inspired by those successes, the same principle has been applied to constructing quantum algorithms.
Variational quantum algorithms (VQAs)~\cite{cerezo2021variational} (including quantum machine learning~\cite{schuld2015introduction}), which optimize parameters of quantum circuits instead of classical neural networks, have emerged as a new method for solving real-world problems.

Despite their promises, the practical usefulness of the VQAs is largely questioned.
One of the main problems is the trainability of quantum circuits.
When parameters are randomly sampled, quantum circuits often have flat cost-function landscapes and cannot be efficiently trained.
This problem, dubbed barren plateaus~\cite{mcclean2018barren}, is expected to prevail among sufficiently expressive quantum circuit ans\"{a}tze~\cite{holmes2022connecting}.
Thus, a deep understanding of barren plateaus is essential for devising an efficient variational algorithm.

For this purpose, a number of studies have suggested quantum circuit ans\"{a}tze without barren plateaus~\cite{pesah2021absence,larocca2022diagnosing,liu2022presence,martin2023barren,barthel2023absence,zhang2023absence}.
However, less is known about how expressive these ans\"{a}tze are. 
Most such circuits are even classically simulable~\cite{cerezo2023does}, and implementing them on quantum hardware is not straightforward~\cite{skolik2023equivariant,east2023all,schatzki2023theoretical}, either.

An alternative (and intuitive) solution is initializing circuits' parameters to provide large gradients.
Indeed, finding good initial parameters is one of the most effective solutions to the vanishing gradient problem in classical neural networks~\cite{glorot2010understanding,he2015delving}.
Likewise, studies have shown that a quantum model can also have large initial gradients when parameters are initialized smartly~\cite{grant2019initialization,jain2021graph,mele2022avoiding,sack2022avoiding,zhang2022gaussian,wang2023trainability,rudolph2023synergistic,park2024hamiltonian}.
However, most of the suggested initialization methods cannot be easily applied to large and deep circuits as they rely on heuristics developed from small circuits~\cite{jain2021graph,mele2022avoiding} or the proven lower bounds of gradient magnitudes are still too small for a deep circuit~\cite{zhang2022gaussian,wang2023trainability}.

In this paper, we propose two novel parameter conditions such that the hardware efficient ansatz (HEA)~\cite{peruzzo2014variational} has large gradients.
Since the HEA utilizes a natural entangling gate provided by the hardware, it is the most suitable quantum circuit ansatz for noisy quantum devices~\cite{preskill2018quantum}. Interestingly, the HEA with a Clifford entangling gate is also preferable to a fault-tolerant quantum computer, as Clifford and single-qubit gates can be implemented with logical qubits rather easily than a parameterized entangling gate~\cite{fujii2015quantum}.
Still, as it is not problem-tailored, the HEA is often expected to be more prone to barren plateaus~\cite{wiersema2020exploring,cerezo2021cost}.
Existing solutions for suppressing barren plateaus in the HEA indeed have certain limitations.
For example, Ref.~\cite{zhang2022gaussian} showed that barren plateaus do not exist for a local observable when parameters are sampled from a Gaussian distribution with a small variance.
However, the authors only considered the HEA with the one-dimensional connectivity, and their lower bound of gradient magnitudes still decays exponentially with the number of qubits for a global observable.

In contrast, the parameter regimes we present here give \textit{constant} gradient magnitudes regardless of circuit depth or the geometry of a circuit.
Our first condition is based on Park and Killoran~\cite{park2024hamiltonian}, which obtained a condition for large gradients when all circuit gates are parameterized.
By removing non-parameterized entangling gates from the HEA using circuit identities, we prove that the same condition applies to the HEA with both local and global observables.
Our second parameter regime is built upon a recent finding~\cite{shtanko2023uncovering} that interprets the HEA to a many-body localized (MBL) system.
We utilize a phenomenological theory of MBL systems~\cite{huse2014phenomenology} to show that the HEA in the MBL phase does not have barren plateaus for a local observable.
Our second parameter condition may have an advantage over the first one as it allows large initial parameters for the diagonal gates.


\begin{figure}
    \centering
    \includegraphics[width=0.98\linewidth]{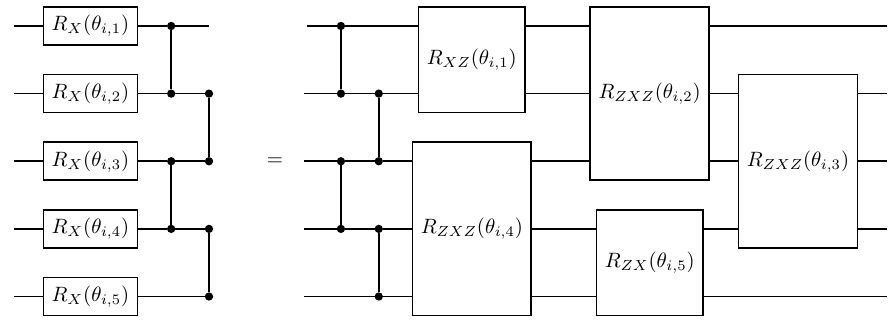}
    \caption{Circuit identity used for removing CZ gates from the HEA. Using the property that the CZ gate is a Clifford gate, we can move CZ gates in each block to the beginning of the block. }
    \label{fig:circ_identity}
\end{figure}

\textit{Hardware efficient ansatz}.--
We consider the HEA for a system with $N$ qubits defined on a finite-dimensional lattice.
Let us define a vector of all parameters $\pmb{\theta}=\{\theta_{i,j}\}$. Then, the output state of the HEA is given by
\begin{align}
    \ket{\psi(\pmb{\theta})} = V(\pmb{\theta}_{p,:}) \cdots V(\pmb{\theta}_{1,:}) \ket{\psi_0},
\end{align}
where $\pmb{\theta}_{i,:} = \{\theta_{i,1},\cdots, \theta_{i,2N}\}$ is a subvector of $\pmb{\theta}$ and $p$ is a parameter determining the total depth of the circuit.
In addition, $V(\pmb{\theta}_{i,:})$ is a unitary operator defined as
\begin{align}
    V(\pmb{\theta}_{i,:}) = \prod_{\braket{j, j'} \in E} W_{j,j'}\prod_{j=1}^N e^{-i Z_j \theta_{i,j+N}/2} \prod_{j=1}^N e^{-i X_j \theta_{i,j}/2}, \label{eq:hea_block}
\end{align}
where $E$ is the set of edges in the lattice, $W_{j,j'}$ is a two-qubit gate between sites $j$ and $j'$, and $\{X_j, Y_j, Z_j\}$ are Pauli operators acting on the $j$-th qubit.
Throughout the paper, we mainly consider $W=\mathrm{CZ} = \mathrm{diag}(1, 1, 1, -1)$, which is a natural entangling gate for major quantum computing platforms~\cite{cirac1995quantum,jaksch2000fast,chow2011simple,figgatt2019parallel,kim2023evidence,evered2023high}.
Still, our arguments can be extended to the HEA with other mutually commuting Clifford entangling gates.

In the VQAs, a cost function is typically given by $C(\pmb{\theta}) = \braket{\psi(\pmb{\theta})|O|\psi(\pmb{\theta})}$ where $O$ is an observable.
A parameterized circuit has barren plateaus with respect to a given parameter distribution $\mathcal{D}$ if $\mathbb{E}_{\pmb{\theta}\sim\mathcal{D}}[(\partial_{ij}C)^2]$ is exponentially small with $N$ for all $i,j$, where $\partial_{ij}C:=\partial C/\partial \theta_{i,j}$.
For completely random parameters, i.e., each $\theta_{i,j}$ is sampled from $\mathcal{U}_{[-\pi, \pi]}$, the HEA has barren plateaus when one of the following conditions is satisfied: (1) $p=\mathrm{poly}(N)$ and $O$ is acting on a constant number of qubits, (2) $p = \Omega(1)$ and $O$ is acting on $\Theta(N)$ sites~\cite{cerezo2021cost}.

\textit{Large gradients with small parameters}.--
Assuming that at least one of the gradient components is constant when all parameters are zero, we prove that a parameter constraint exists such that the gradient can be bounded below by a constant.

\begin{theorem} \label{thm:large_gradients_constraints}
    Let $C(\pmb{\theta}) = \braket{\psi(\pmb{\theta})|O|\psi(\pmb{\theta})}$ be the cost function where $O$ is either a Pauli string or $k$-local Hamiltonian.
    Suppose that there exist $n,m$ such that $|\partial_{n,m} C |_{\pmb{\theta}=0} = \Omega(1)$.
    Then, there exists a constant $\gamma > 0$ such that $|\partial_{n,m} C| =\Omega(1)$ when $0 \leq \theta_{i,j} \leq \gamma / (pN)$ is satisfied for all $i$ and $j$.
\end{theorem}
See Appendix~\ref{app:small_param_constant_grad} for a proof. 
Theorem~\ref{thm:large_gradients_constraints} is based on Ref.~\cite{park2024hamiltonian}, which proved the existence of a parameter condition such that a circuit has a constant gradient component when all gates are parameterized.
We use a circuit identity given in Fig.~\ref{fig:circ_identity} to translate the HEA into a circuit without non-parameterized entangling gates.
The circuit identity enables us to move the CZ layer in the $2i$-th block to the beginning of the block, which cancels the CZ layer in the $2i-1$-th block.
The resulting circuit only has parameterized gates (for even $p$) generated by, at most, $k$-local operators (acting on at most $k$ nearby sites), where $k$ is determined by the connectivity of the original circuit.
This procedure recovers a setup used in Ref.~\cite{park2024hamiltonian}.
The same argument also works for odd $p$, which remains a single layer of CZ gates acting on the initial state (see Appendix~\ref{app:small_param_constant_grad}).
In addition, we note that one can easily find a product state $\ket{\psi_0}$ satisfying $|\partial_{n,m} C |_{\pmb{\theta}=0} = \Omega(1)$ when $O$ is a Pauli string.

Theorem~\ref{thm:large_gradients_constraints} can be compared to the main result of Ref.~\cite{zhang2022gaussian}, which proved that the magnitudes of the gradient could be lower bounded by $\Theta[(pS)^{-S}]$ when the parameters are drawn from $\mathcal{N}(0, [1/\sqrt{4pS}]^2)$ (in our notation).
Here, $S$ is the weight of the observable $O$, which counts the number of qubits that $O$ acts non-trivially (e.g., $X_1Z_5Z_7$ has $S=3$).
Three major differences are listed as follows.
First, the parameters are smaller for Theorem~\ref{thm:large_gradients_constraints}, $\Theta[1/(pN)]$ versus $\Theta(1/\sqrt{pN})$, but the lower bound is much bigger, $\Theta(1)$ versus $\Theta[(pS)^{-S}]$.
Second, Ref.~\cite{zhang2022gaussian} only considered the 1D HEA, whereas our theorem applies to any finite-dimensional lattices, including the heavy-hexagon lattice upon which IBM's recent quantum processors are implemented~\cite{heavy-hex-lattice}.
Finally, Theorem~\ref{thm:large_gradients_constraints} allows additional gates applied to the initial state, such as data-encoding gates widely used in quantum machine learning setup, as long as the circuit has large gradients when all trainable parameters are zero (see Appendix~\ref{app:machine_learning}).

\textit{Floquet-MBL initialization for a large gradient}.--
One of the potential limitations of the previous parameter condition is that the parameters are too small when applied to a circuit with a large depth.
In this case, the parameters between each instance are nearly the same, and an advantage of the randomness in initial parameters~\cite{wessels1992avoiding} may be lost.

Our second parameter condition overcomes this problem by allowing the parameters for the RZ gates to be random in $\mathcal{U}_{[-\pi, \pi]}$.
Formally, the parameter condition is written as follows:
\begin{gather}
    \vartheta_i = \theta_{i,j} \text{ for all } 1 \leq j \leq N \text{ and } 0 \leq \vartheta_i \leq \vartheta_{c} \text{ for all } i, \nonumber \\
    \quad \theta_{i,j} \sim \mathcal{U}_{[-\pi, \pi]} \text{ for all } N+1 \leq j  \leq 2N, \label{eq:mbl_cond}
\end{gather}
where $\vartheta_{c}$ is the critical point between the chaotic and the MBL phases. For the 1D HEA, we find $0.13 \lesssim \vartheta_{c} \lesssim 0.16$ (see Appendix~\ref{app:floquet-mbl}).

Our circuit is in the MBL phase when the parameters satisfy this condition.
A phenomenological theory of the MBL~\cite{huse2014phenomenology} suggests that one can find a Hamiltonian $H_{\rm MBL}$ such that
\begin{align}
    V(\pmb{\theta}_{k,:}) \cdots V(\pmb{\theta}_{1,:}) = e^{-i H_{\rm MBL} kT},
\end{align}
for any $k \geq 1$ and a constant $T$, where $H_{\rm MBL}$ can be written as
\begin{align}
    H_{\rm MBL} &= \sum_{i=1}^N J_i \tau_{i}^z +  \sum_{i \neq j} J_{ij} \tau_{i}^z \tau_{j}^z \nonumber \\
    & \qquad +\sum_{\text{all distinct }i,j,k } J_{ijk} \tau_{i}^z \tau_{j}^z \tau_{k}^z +\cdots. \label{eq:mbl_phenom}
\end{align}
Here, $\tau_{i}^z$ is a local integral of motion, which has a finite overlap with $Z_i$.

Let us consider the gradient component for $\theta_{p,1}$ when $O=Y_1$ and $\ket{\psi_0} = \ket{0^N}$ is used. From the definition of the cost function, we obtain
\begin{align}
    \frac{\partial C}{\partial \theta_{p,1}} = \frac{i}{2} \braket{0^N | U_{[1:p-1]}^\dagger [X_1, \widetilde{Y_1}] U_{[1:p-1]} | 0^N},
\end{align}
where $U_{[1:p-1]}=V(\pmb{\theta}_{p-1,:})\cdots V(\pmb{\theta}_{1,:})$ is a subcircuit of the HEA and $\widetilde{Y_1} := V(\pmb{\theta}_{p,:})^\dagger Y_1 V(\pmb{\theta}_{p,:})$.
After some steps, the following expression is obtained:
\begin{align*}
    [X_1, \widetilde{Y_1}] &= 2i \cos(\theta_{p,N+1}) [\cos(\vartheta_p) Z_1 + \sin(\vartheta_p) Y_1] \\
    & \quad \times \prod_{j \in \mathcal{N}(1)} [\cos(\vartheta_p) Z_i + \sin(\vartheta_p) Y_i], \numberthis
\end{align*}
where $\mathcal{N}(i)$ is a set of all neighbors of $i$ in a given lattice (see Appendix~\ref{app:grad_scaling_mbl} for details).

\begin{figure}[t]
    \centering
    \includegraphics[width=0.85\linewidth]{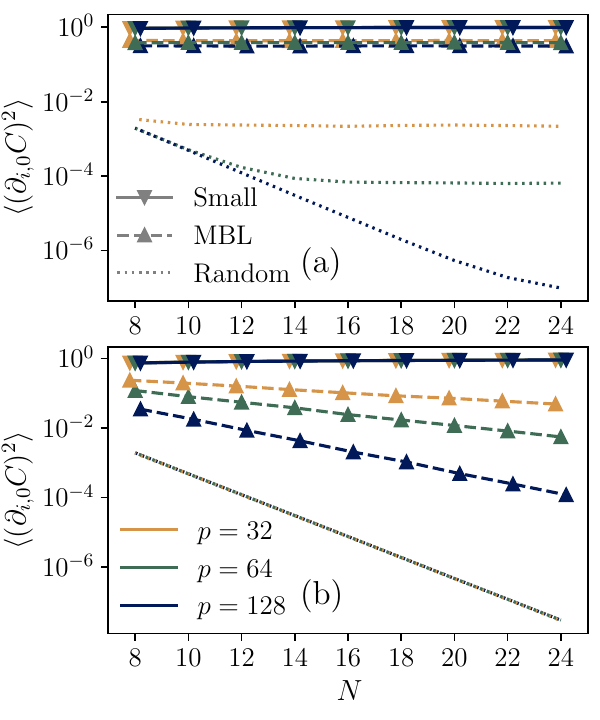}
    \caption{Averaged squared gradients as functions of $N$ for $p \in [32, 64, 128]$. Observables (a) $O=Y_1$ and (b) $O=Y_1 \prod_{j =2}^N Z_j$ are used. Each data point presents the averaged gradient components for the RX gate acting on the first qubit, $\sum_{i=1}^p (\partial_{i,0}C)^2/p$.
    For each parameter initialization scheme, results are averaged over $2^{10}$ randomly sampled parameters.
    For the Small initialization, the gradient magnitudes do not decay with $N$ regardless of the observable.
    On the other hand, the MBL initialization shows $\Theta(1)$ gradient magnitudes when a local observable is used, whereas they decay exponentially for a global observable.
    }
    \label{fig:grads_compare_1d}
\end{figure}

In summary, the gradient is expressed as the sum of multi-point correlation functions.
From the fact that $Z_i$ has a finite overlap with $\tau^i_z$, and the correlation functions involving Pauli-$Y$ operators are relatively small in the MBL systems~\cite{serbyn2014quantum}, we obtain
\begin{align}
    \frac{\partial C}{\partial \theta_{p,N+1}} \approx - \cos(\theta_{p,N+1}) \bigl[ A^2\cos(\vartheta_p) \bigr]^{1+|\mathcal{N}(1)|}
\end{align}
for sufficiently large $p$ (see Appendix~\ref{app:grad_scaling_mbl} for details). 
Here, $A=\Tr[\tau^i_z Z_i]/2^N$ quantifies the overlap between two operators and is independent of $N$ and $p$.
Hence, we obtain $\mathbb{E}_{\pmb{\theta}}[(\partial C/\partial \theta_{p,N+1})^2] \approx \bigl[ A^2\cos(\vartheta_p) \bigr]^{1+|\mathcal{N}(1)|}/2$.
This implies that the HEA in the MBL phase does not have barren plateaus in any depth for this observable.
One can also repeat the same calculation for other observables.
For a global observable $O=Y_1\prod_{j=2}^N Z_j$, we obtain $\mathbb{E}_{\pmb{\theta}}[(\partial C/\partial \theta_{p,N+1})^2] \approx \bigl[ A^2\cos(\vartheta_p) \bigr]^{N-|\mathcal{N}(1)|}/2$, which decays exponentially with $N$ (see Appendix~\ref{app:grad_scaling_mbl} for details).

There is a subtlety in applying our argument here to the HEA in a higher dimensional lattice.
Recent studies have claimed that the MBL phase does not exist in the thermodynamic limit when the dimension is larger than one (see, e.g., Ref.~\cite{thiery2018many}).
However, as one can still observe a signature of the MBL transition for a finite-size system~\cite{wahl2019signatures},
we expect that our MBL parameter condition would work even in a higher dimensional HEA for system sizes tractable to intermediate-scale quantum computers.

\textit{Numerical simulations}.--
We numerically compare the magnitudes of the gradients when parameters are randomly sampled from the following distributions. (1) \textbf{Small}: All parameters are drawn from $\mathcal{U}_{[0, \pi/(pN)]}$, (2) \textbf{MBL}: Parameters follow Eq.~\eqref{eq:mbl_cond} with $\vartheta_i \in \mathcal{U}_{[0, 0.1]}$, and (3) \textbf{Random}: All parameters are completely random, i.e., $\theta_{i,j} \sim \mathcal{U}_{[0, 2\pi]}$ for all $i,j$.

We use two observables $O=Y_1$ and $O=Y_1 \prod_{j=2}^N Z_j$, and the initial state given by $\ket{\psi_0}=\ket{0^N}$.
Simple computation gives that $\partial_{i,0}C|_{\pmb{\theta} = 0} = -1$ for both the observables regardless of $i$.
From Theorem~\ref{thm:large_gradients_constraints}, we expect that our first parameter scheme (Small) could give large gradients~\footnote{The value of $\gamma$ obtained from the proof of Theorem~\ref{thm:large_gradients_constraints} can be smaller than $\pi$. However, since many inequalities we have used for the proof are not tight, we regard Theorem~\ref{thm:large_gradients_constraints} as a rough estimation of the order of such a parameter condition. The exact scaling of the parameter condition can be further investigated numerically. See also Appendix~\ref{app:additional_numerical_results}.}.
On the other hand, the MBL parameter scheme will give a constant magnitude of $\partial_{p,0}C$ for $O=Y_1$, whereas the magnitude would decay exponentially with $N$ for sufficiently large $p$ when $O=Y_1 \prod_{j=1}^N Z_j$ is used.

The scaling behaviors of gradients for $O=Y_1$ from these schemes are plotted in Fig.~\ref{fig:grads_compare_1d}(a).
We observe that gradient magnitudes do not decay with $N$ when parameters are initialized following the small or MBL scheme, which is consistent with our theoretical investigations.
On the other hand, when parameters are completely random, gradient magnitudes decay exponentially with $N$ for small $N$, and they saturate after $N=N_0(p)$.
This is consistent with previous observations in Refs.~\cite{mcclean2018barren,cerezo2021cost,park2024hamiltonian}.

We also plot results for $O=Y_1 \prod_{j=1}^N Z_j$ in Fig.~\ref{fig:grads_compare_1d}(b), which shows that the magnitudes of the gradients are constant for the small parameter scheme, whereas they decay exponentially with $N$ for the MBL distribution.
These also agree with our theoretical expectations.
In addition, the random parameter scheme gives exponential decay of the gradient magnitudes regardless of $p$.
This is also consistent with Ref.~\cite{cerezo2021cost}, which proved that barren plateaus appear in $\Omega(1)$ depth for a global observable. 
Lastly, while both the MBL and random parameter schemes show exponentially decaying gradients, the decaying rate is much lower for the MBL scheme.
This suggests that the MBL parameter scheme may still give practical advantages even for a global observable when $p$ and $N$ are not too large.

In Appendix~\ref{app:additional_numerical_results}, we present additional numerical results for the 2D HEA. We also compare our parameter conditions and the random Gaussian initialization suggested in Ref.~\cite{zhang2023absence}.

\begin{figure}
    \centering
    \includegraphics[width=0.85\linewidth]{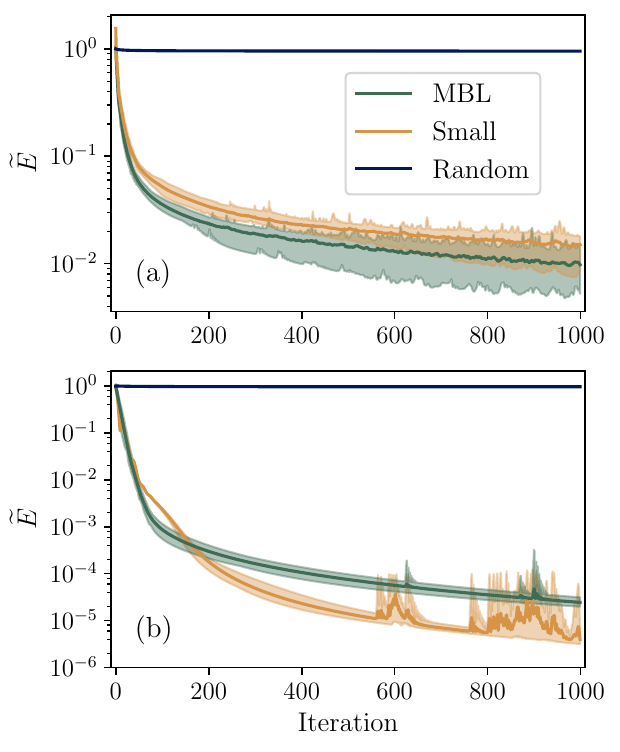}
    \caption{Normalized energies $\widetilde{E} = (\braket{H_{1,2}} - E_{\rm GS})/|E_{\rm GS}|$ as functions of optimization steps for (a) the Heisenberg model ($H_1$) and (b) the cluster model ($H_2$) with external fields. The HEA with $N=20$ and $p=256$ is used. We optimize the parameters using Adam~\cite{kingma2014adam} with learning rates (a) $\eta = 0.005$ and (b) $\eta = 0.001$, which are chosen from hyperparameter optimizations. For each initialization scheme, we run $16$ independent VQE instances. Solid curves show the averaged values for each step, while the shaded regions indicate the range between the worst and best-performing instances.}
    \label{fig:vqe_learning_curves}
\end{figure}

\textit{Solving quantum many-body Hamiltonians}.--
We now solve the ground state problem of two Hamiltonians by simulating variational quantum eigensolvers (VQEs).
We consider the one-dimensional Heisenberg and cluster models with external fields given by
\begin{align}
    H_1 &= \sum_{i=1}^{N-1} \Bigl[ X_i X_{i+1} + Y_i Y_{i+1} + Z_i Z_{i+1} \Bigr] + h_1 \sum_{i=1}^N Z_i \\
    H_2 &= -\sum_{i=2}^{N-2} Z_{i-1} X_{i} Z_{i+1}  - X_1 Z_2 - Z_{N-1} X_N - h_2 \sum_{i=1}^N Z_i,
\end{align}
with the strength of external fields $h_1=h_2=1$.

We use the cost functions given by the expectation value of the Hamiltonians for the output states of circuits (i.e., $C=\braket{H_{1,2}}$) and optimize the parameters using Adam~\cite{kingma2014adam} with exactly computed gradients.
We choose $\ket{\psi_0} = \ket{y;+}^{\otimes N}$ as the initial state of the circuit since the gradient components for RX gates are $\Omega(1)$ when all parameters are zero with this choice, i.e., $|\partial_{n,m}C|_{\pmb{\theta}=0} = \Omega(1)$ for all $n\in \{1, \cdots, N\}$ and $1 \leq m \leq N$.

For the 1D HEA with $N=20$ and $p=256$, the learning curves are plotted in Fig.~\ref{fig:vqe_learning_curves}. We choose such relatively large $p$ to ensure that gradients are sufficiently small when parameters are completely random.
The results show that the circuits initialized following our parameter schemes, Small and MBL, provide much better convergence than Random.
However, the best initialization methods depend on the Hamiltonians: The MBL initialization outperforms the Small initialization for $H_1$, but the opposite holds for $H_2$.
We also observe that the HEA finds the nearly exact ground state for $H_2$ (with $\widetilde{E}\approx 3 \times 10^{-6}$), whereas the results for $H_1$ are relatively poor.
These differences are from the expressivity of the HEA for the target problems.
The ground state of $H_1$ lies within the subspace $J_z=0$ where $J_z=\sum_{i=1}^N Z_i$ is the total spin operator.
However, the HEA cannot capture this symmetry, and the output state always overlaps with subspaces with other $J_z$ values.
On the other hand, the HEA is a natural ansatz~\cite{wecker2015progress,hadfield2019quantum,park2021efficient} for $H_2$, which is from our circuit identity (Fig.~\ref{fig:circ_identity}).
These results indicate that our parameter initialization sufficiently avoids barren plateaus, and the expressivity of the circuit or local minima can be more critical problems determining the performance of quantum variational algorithms.

In Appendix~\ref{app:machine_learning}, we present additional results solving a machine-learning problem using the HEA. We show that our parameter initialization schemes, Small and MBL, offer better performance also for a binary classification task, a typical supervised learning problem.

\textit{Conclusion and outlook}.--
We found two novel parameter conditions where the hardware-efficient ansatz (HEA) does not have barren plateaus.
In contrast to other known conditions, our ones provide a constant gradient regardless of circuit depth.

In addition to the practical advantage of initializing the HEA, our second parameter condition can be a counterexample to a recent claim~\cite{cerezo2023does} that all barren plateaus free ans\"{a}tze are classically simulable~\footnote{In contrast, the HEA within the first parameter condition is classically simulable up to an inverse polynomial additive error using algorithms developed in Refs.~\cite{bravyi2021classical,coble2022quasi}. However, we expect that one can extend Theorem~\ref{thm:large_gradients_constraints} to an arbitrary graph (instead of a $D$-dimensional lattice), which makes those algorithms cannot be applied.}.
This is because there is no known classically efficient algorithm for simulating the MLB system for an exponentially long time, and our argument based on the phenomenological theory of the MBL systems~\cite{huse2014phenomenology} guarantees a gradient component with a constant magnitude even at an exponentially long time.
Still, rigorous proof of this argument should be addressed in future work as it requires careful complexity analysis.

Our MBL parameter condition also raises an interesting question on the role of entanglement in barren plateaus. It is often assumed that entanglement volume-law implies the onset of barren plateaus~\cite{marrero2021entanglement,sack2022avoiding,leone2022practical}.
However, our results suggest that it is not always the case as a long-time-evolved state in the MBL system follows the volume-law of entanglement~\cite{bardarson2012unbounded}, whereas our argument guarantees a constant magnitude gradient component at any time.

\textit{Note added.--}
In the days prior to the submission of this manuscript, a preprint~\cite{shi2024avoiding} proposing an initialization scheme for the HEA that provides $\Theta(1)$ gradient norm is uploaded. However, their bound is only proved for the HEA with the one-dimensional connectivity.

\textit{Acknowledgements}.--
CYP thanks Kyunghyun Baek for the initial discussion and Joseph Bowles, Nathan Killoran, Korbinian Kottmann, and Maria Schuld for helpful comments.
This research used resources of the National Energy Research Scientific Computing Center, a DOE Office of Science User Facility supported by the Office of Science of the U.S. Department of Energy under Contract No. DE-AC02-05CH11231 using NERSC award NERSC DDR-ERCAP0025705.
This work was partly supported by the Basic Science Research Program through the National Research Foundation of Korea (NRF), funded by the Ministry of Education, Science and Technology (NRF-2022M3H3A106307411, NRF-2023M3K5A1094805, and NRF-2023M3K5A109481311) and Institute for Information \& communications Technology Promotion (IITP) grant funded by the Korea government(MSIP) (No. 2019-0-00003, Research and Development of Core technologies for Programming, Running, Implementing and Validating of Fault-Tolerant Quantum Computing System). 
JH acknowledges Xanadu for hosting his sabbatical year visit.
Numerical simulations were performed using \textsc{PennyLane}~\cite{bergholm2018pennylane} software package with \textsc{Lightning}~\cite{Lightning} plugin.

\medskip
\bibliography{references.bib}

\begin{thebibliography}{74}%
\makeatletter
\providecommand \@ifxundefined [1]{%
 \@ifx{#1\undefined}
}%
\providecommand \@ifnum [1]{%
 \ifnum #1\expandafter \@firstoftwo
 \else \expandafter \@secondoftwo
 \fi
}%
\providecommand \@ifx [1]{%
 \ifx #1\expandafter \@firstoftwo
 \else \expandafter \@secondoftwo
 \fi
}%
\providecommand \natexlab [1]{#1}%
\providecommand \enquote  [1]{``#1''}%
\providecommand \bibnamefont  [1]{#1}%
\providecommand \bibfnamefont [1]{#1}%
\providecommand \citenamefont [1]{#1}%
\providecommand \href@noop [0]{\@secondoftwo}%
\providecommand \href [0]{\begingroup \@sanitize@url \@href}%
\providecommand \@href[1]{\@@startlink{#1}\@@href}%
\providecommand \@@href[1]{\endgroup#1\@@endlink}%
\providecommand \@sanitize@url [0]{\catcode `\\12\catcode `\$12\catcode
  `\&12\catcode `\#12\catcode `\^12\catcode `\_12\catcode `\%12\relax}%
\providecommand \@@startlink[1]{}%
\providecommand \@@endlink[0]{}%
\providecommand \url  [0]{\begingroup\@sanitize@url \@url }%
\providecommand \@url [1]{\endgroup\@href {#1}{\urlprefix }}%
\providecommand \urlprefix  [0]{URL }%
\providecommand \Eprint [0]{\href }%
\providecommand \doibase [0]{https://doi.org/}%
\providecommand \selectlanguage [0]{\@gobble}%
\providecommand \bibinfo  [0]{\@secondoftwo}%
\providecommand \bibfield  [0]{\@secondoftwo}%
\providecommand \translation [1]{[#1]}%
\providecommand \BibitemOpen [0]{}%
\providecommand \bibitemStop [0]{}%
\providecommand \bibitemNoStop [0]{.\EOS\space}%
\providecommand \EOS [0]{\spacefactor3000\relax}%
\providecommand \BibitemShut  [1]{\csname bibitem#1\endcsname}%
\let\auto@bib@innerbib\@empty
\bibitem [{\citenamefont {Krizhevsky}\ \emph {et~al.}(2017)\citenamefont
  {Krizhevsky}, \citenamefont {Sutskever},\ and\ \citenamefont
  {Hinton}}]{krizhevsky2017imagenet}%
  \BibitemOpen
  \bibfield  {author} {\bibinfo {author} {\bibfnamefont {A.}~\bibnamefont
  {Krizhevsky}}, \bibinfo {author} {\bibfnamefont {I.}~\bibnamefont
  {Sutskever}},\ and\ \bibinfo {author} {\bibfnamefont {G.~E.}\ \bibnamefont
  {Hinton}},\ }\bibfield  {title} {\bibinfo {title} {Imagenet classification
  with deep convolutional neural networks},\ }\href@noop {} {\bibfield
  {journal} {\bibinfo  {journal} {Communications of the ACM}\ }\textbf
  {\bibinfo {volume} {60}},\ \bibinfo {pages} {84} (\bibinfo {year}
  {2017})}\BibitemShut {NoStop}%
\bibitem [{\citenamefont {Silver}\ \emph {et~al.}(2016)\citenamefont {Silver},
  \citenamefont {Huang}, \citenamefont {Maddison}, \citenamefont {Guez},
  \citenamefont {Sifre}, \citenamefont {Van Den~Driessche}, \citenamefont
  {Schrittwieser}, \citenamefont {Antonoglou}, \citenamefont {Panneershelvam},
  \citenamefont {Lanctot} \emph {et~al.}}]{silver2016mastering}%
  \BibitemOpen
  \bibfield  {author} {\bibinfo {author} {\bibfnamefont {D.}~\bibnamefont
  {Silver}}, \bibinfo {author} {\bibfnamefont {A.}~\bibnamefont {Huang}},
  \bibinfo {author} {\bibfnamefont {C.~J.}\ \bibnamefont {Maddison}}, \bibinfo
  {author} {\bibfnamefont {A.}~\bibnamefont {Guez}}, \bibinfo {author}
  {\bibfnamefont {L.}~\bibnamefont {Sifre}}, \bibinfo {author} {\bibfnamefont
  {G.}~\bibnamefont {Van Den~Driessche}}, \bibinfo {author} {\bibfnamefont
  {J.}~\bibnamefont {Schrittwieser}}, \bibinfo {author} {\bibfnamefont
  {I.}~\bibnamefont {Antonoglou}}, \bibinfo {author} {\bibfnamefont
  {V.}~\bibnamefont {Panneershelvam}}, \bibinfo {author} {\bibfnamefont
  {M.}~\bibnamefont {Lanctot}}, \emph {et~al.},\ }\bibfield  {title} {\bibinfo
  {title} {Mastering the game of go with deep neural networks and tree
  search},\ }\href@noop {} {\bibfield  {journal} {\bibinfo  {journal} {Nature}\
  }\textbf {\bibinfo {volume} {529}},\ \bibinfo {pages} {484} (\bibinfo {year}
  {2016})}\BibitemShut {NoStop}%
\bibitem [{\citenamefont {Brown}\ \emph {et~al.}(2020)\citenamefont {Brown},
  \citenamefont {Mann}, \citenamefont {Ryder}, \citenamefont {Subbiah},
  \citenamefont {Kaplan}, \citenamefont {Dhariwal}, \citenamefont
  {Neelakantan}, \citenamefont {Shyam}, \citenamefont {Sastry}, \citenamefont
  {Askell} \emph {et~al.}}]{brown2020language}%
  \BibitemOpen
  \bibfield  {author} {\bibinfo {author} {\bibfnamefont {T.}~\bibnamefont
  {Brown}}, \bibinfo {author} {\bibfnamefont {B.}~\bibnamefont {Mann}},
  \bibinfo {author} {\bibfnamefont {N.}~\bibnamefont {Ryder}}, \bibinfo
  {author} {\bibfnamefont {M.}~\bibnamefont {Subbiah}}, \bibinfo {author}
  {\bibfnamefont {J.~D.}\ \bibnamefont {Kaplan}}, \bibinfo {author}
  {\bibfnamefont {P.}~\bibnamefont {Dhariwal}}, \bibinfo {author}
  {\bibfnamefont {A.}~\bibnamefont {Neelakantan}}, \bibinfo {author}
  {\bibfnamefont {P.}~\bibnamefont {Shyam}}, \bibinfo {author} {\bibfnamefont
  {G.}~\bibnamefont {Sastry}}, \bibinfo {author} {\bibfnamefont
  {A.}~\bibnamefont {Askell}}, \emph {et~al.},\ }\bibfield  {title} {\bibinfo
  {title} {Language models are few-shot learners},\ }in\ \href@noop {} {\emph
  {\bibinfo {booktitle} {Advances in neural information processing systems}}},\
  Vol.~\bibinfo {volume} {33}\ (\bibinfo {year} {2020})\ pp.\ \bibinfo {pages}
  {1877--1901}\BibitemShut {NoStop}%
\bibitem [{\citenamefont {Jumper}\ \emph {et~al.}(2021)\citenamefont {Jumper},
  \citenamefont {Evans}, \citenamefont {Pritzel}, \citenamefont {Green},
  \citenamefont {Figurnov}, \citenamefont {Ronneberger}, \citenamefont
  {Tunyasuvunakool}, \citenamefont {Bates}, \citenamefont {{\v{Z}}{\'\i}dek},
  \citenamefont {Potapenko} \emph {et~al.}}]{jumper2021highly}%
  \BibitemOpen
  \bibfield  {author} {\bibinfo {author} {\bibfnamefont {J.}~\bibnamefont
  {Jumper}}, \bibinfo {author} {\bibfnamefont {R.}~\bibnamefont {Evans}},
  \bibinfo {author} {\bibfnamefont {A.}~\bibnamefont {Pritzel}}, \bibinfo
  {author} {\bibfnamefont {T.}~\bibnamefont {Green}}, \bibinfo {author}
  {\bibfnamefont {M.}~\bibnamefont {Figurnov}}, \bibinfo {author}
  {\bibfnamefont {O.}~\bibnamefont {Ronneberger}}, \bibinfo {author}
  {\bibfnamefont {K.}~\bibnamefont {Tunyasuvunakool}}, \bibinfo {author}
  {\bibfnamefont {R.}~\bibnamefont {Bates}}, \bibinfo {author} {\bibfnamefont
  {A.}~\bibnamefont {{\v{Z}}{\'\i}dek}}, \bibinfo {author} {\bibfnamefont
  {A.}~\bibnamefont {Potapenko}}, \emph {et~al.},\ }\bibfield  {title}
  {\bibinfo {title} {Highly accurate protein structure prediction with
  alphafold},\ }\href@noop {} {\bibfield  {journal} {\bibinfo  {journal}
  {Nature}\ }\textbf {\bibinfo {volume} {596}},\ \bibinfo {pages} {583}
  (\bibinfo {year} {2021})}\BibitemShut {NoStop}%
\bibitem [{\citenamefont {Cerezo}\ \emph
  {et~al.}(2021{\natexlab{a}})\citenamefont {Cerezo}, \citenamefont
  {Arrasmith}, \citenamefont {Babbush}, \citenamefont {Benjamin}, \citenamefont
  {Endo}, \citenamefont {Fujii}, \citenamefont {McClean}, \citenamefont
  {Mitarai}, \citenamefont {Yuan}, \citenamefont {Cincio} \emph
  {et~al.}}]{cerezo2021variational}%
  \BibitemOpen
  \bibfield  {author} {\bibinfo {author} {\bibfnamefont {M.}~\bibnamefont
  {Cerezo}}, \bibinfo {author} {\bibfnamefont {A.}~\bibnamefont {Arrasmith}},
  \bibinfo {author} {\bibfnamefont {R.}~\bibnamefont {Babbush}}, \bibinfo
  {author} {\bibfnamefont {S.~C.}\ \bibnamefont {Benjamin}}, \bibinfo {author}
  {\bibfnamefont {S.}~\bibnamefont {Endo}}, \bibinfo {author} {\bibfnamefont
  {K.}~\bibnamefont {Fujii}}, \bibinfo {author} {\bibfnamefont {J.~R.}\
  \bibnamefont {McClean}}, \bibinfo {author} {\bibfnamefont {K.}~\bibnamefont
  {Mitarai}}, \bibinfo {author} {\bibfnamefont {X.}~\bibnamefont {Yuan}},
  \bibinfo {author} {\bibfnamefont {L.}~\bibnamefont {Cincio}}, \emph
  {et~al.},\ }\bibfield  {title} {\bibinfo {title} {Variational quantum
  algorithms},\ }\href@noop {} {\bibfield  {journal} {\bibinfo  {journal} {Nat.
  Rev. Phys.}\ }\textbf {\bibinfo {volume} {3}},\ \bibinfo {pages} {625}
  (\bibinfo {year} {2021}{\natexlab{a}})}\BibitemShut {NoStop}%
\bibitem [{\citenamefont {Schuld}\ \emph {et~al.}(2015)\citenamefont {Schuld},
  \citenamefont {Sinayskiy},\ and\ \citenamefont
  {Petruccione}}]{schuld2015introduction}%
  \BibitemOpen
  \bibfield  {author} {\bibinfo {author} {\bibfnamefont {M.}~\bibnamefont
  {Schuld}}, \bibinfo {author} {\bibfnamefont {I.}~\bibnamefont {Sinayskiy}},\
  and\ \bibinfo {author} {\bibfnamefont {F.}~\bibnamefont {Petruccione}},\
  }\bibfield  {title} {\bibinfo {title} {An introduction to quantum machine
  learning},\ }\href@noop {} {\bibfield  {journal} {\bibinfo  {journal}
  {Contemporary Physics}\ }\textbf {\bibinfo {volume} {56}},\ \bibinfo {pages}
  {172} (\bibinfo {year} {2015})}\BibitemShut {NoStop}%
\bibitem [{\citenamefont {McClean}\ \emph {et~al.}(2018)\citenamefont
  {McClean}, \citenamefont {Boixo}, \citenamefont {Smelyanskiy}, \citenamefont
  {Babbush},\ and\ \citenamefont {Neven}}]{mcclean2018barren}%
  \BibitemOpen
  \bibfield  {author} {\bibinfo {author} {\bibfnamefont {J.~R.}\ \bibnamefont
  {McClean}}, \bibinfo {author} {\bibfnamefont {S.}~\bibnamefont {Boixo}},
  \bibinfo {author} {\bibfnamefont {V.~N.}\ \bibnamefont {Smelyanskiy}},
  \bibinfo {author} {\bibfnamefont {R.}~\bibnamefont {Babbush}},\ and\ \bibinfo
  {author} {\bibfnamefont {H.}~\bibnamefont {Neven}},\ }\bibfield  {title}
  {\bibinfo {title} {Barren plateaus in quantum neural network training
  landscapes},\ }\href@noop {} {\bibfield  {journal} {\bibinfo  {journal} {Nat.
  Commun.}\ }\textbf {\bibinfo {volume} {9}},\ \bibinfo {pages} {1} (\bibinfo
  {year} {2018})}\BibitemShut {NoStop}%
\bibitem [{\citenamefont {Holmes}\ \emph {et~al.}(2022)\citenamefont {Holmes},
  \citenamefont {Sharma}, \citenamefont {Cerezo},\ and\ \citenamefont
  {Coles}}]{holmes2022connecting}%
  \BibitemOpen
  \bibfield  {author} {\bibinfo {author} {\bibfnamefont {Z.}~\bibnamefont
  {Holmes}}, \bibinfo {author} {\bibfnamefont {K.}~\bibnamefont {Sharma}},
  \bibinfo {author} {\bibfnamefont {M.}~\bibnamefont {Cerezo}},\ and\ \bibinfo
  {author} {\bibfnamefont {P.~J.}\ \bibnamefont {Coles}},\ }\bibfield  {title}
  {\bibinfo {title} {Connecting ansatz expressibility to gradient magnitudes
  and barren plateaus},\ }\href@noop {} {\bibfield  {journal} {\bibinfo
  {journal} {PRX Quantum}\ }\textbf {\bibinfo {volume} {3}},\ \bibinfo {pages}
  {010313} (\bibinfo {year} {2022})}\BibitemShut {NoStop}%
\bibitem [{\citenamefont {Pesah}\ \emph {et~al.}(2021)\citenamefont {Pesah},
  \citenamefont {Cerezo}, \citenamefont {Wang}, \citenamefont {Volkoff},
  \citenamefont {Sornborger},\ and\ \citenamefont {Coles}}]{pesah2021absence}%
  \BibitemOpen
  \bibfield  {author} {\bibinfo {author} {\bibfnamefont {A.}~\bibnamefont
  {Pesah}}, \bibinfo {author} {\bibfnamefont {M.}~\bibnamefont {Cerezo}},
  \bibinfo {author} {\bibfnamefont {S.}~\bibnamefont {Wang}}, \bibinfo {author}
  {\bibfnamefont {T.}~\bibnamefont {Volkoff}}, \bibinfo {author} {\bibfnamefont
  {A.~T.}\ \bibnamefont {Sornborger}},\ and\ \bibinfo {author} {\bibfnamefont
  {P.~J.}\ \bibnamefont {Coles}},\ }\bibfield  {title} {\bibinfo {title}
  {Absence of barren plateaus in quantum convolutional neural networks},\
  }\href@noop {} {\bibfield  {journal} {\bibinfo  {journal} {Phys. Rev. X}\
  }\textbf {\bibinfo {volume} {11}},\ \bibinfo {pages} {041011} (\bibinfo
  {year} {2021})}\BibitemShut {NoStop}%
\bibitem [{\citenamefont {Larocca}\ \emph {et~al.}(2022)\citenamefont
  {Larocca}, \citenamefont {Czarnik}, \citenamefont {Sharma}, \citenamefont
  {Muraleedharan}, \citenamefont {Coles},\ and\ \citenamefont
  {Cerezo}}]{larocca2022diagnosing}%
  \BibitemOpen
  \bibfield  {author} {\bibinfo {author} {\bibfnamefont {M.}~\bibnamefont
  {Larocca}}, \bibinfo {author} {\bibfnamefont {P.}~\bibnamefont {Czarnik}},
  \bibinfo {author} {\bibfnamefont {K.}~\bibnamefont {Sharma}}, \bibinfo
  {author} {\bibfnamefont {G.}~\bibnamefont {Muraleedharan}}, \bibinfo {author}
  {\bibfnamefont {P.~J.}\ \bibnamefont {Coles}},\ and\ \bibinfo {author}
  {\bibfnamefont {M.}~\bibnamefont {Cerezo}},\ }\bibfield  {title} {\bibinfo
  {title} {Diagnosing barren plateaus with tools from quantum optimal
  control},\ }\href@noop {} {\bibfield  {journal} {\bibinfo  {journal}
  {Quantum}\ }\textbf {\bibinfo {volume} {6}},\ \bibinfo {pages} {824}
  (\bibinfo {year} {2022})}\BibitemShut {NoStop}%
\bibitem [{\citenamefont {Liu}\ \emph {et~al.}(2022)\citenamefont {Liu},
  \citenamefont {Yu}, \citenamefont {Duan},\ and\ \citenamefont
  {Deng}}]{liu2022presence}%
  \BibitemOpen
  \bibfield  {author} {\bibinfo {author} {\bibfnamefont {Z.}~\bibnamefont
  {Liu}}, \bibinfo {author} {\bibfnamefont {L.-W.}\ \bibnamefont {Yu}},
  \bibinfo {author} {\bibfnamefont {L.-M.}\ \bibnamefont {Duan}},\ and\
  \bibinfo {author} {\bibfnamefont {D.-L.}\ \bibnamefont {Deng}},\ }\bibfield
  {title} {\bibinfo {title} {Presence and absence of barren plateaus in
  tensor-network based machine learning},\ }\href@noop {} {\bibfield  {journal}
  {\bibinfo  {journal} {Phys. Rev. Lett.}\ }\textbf {\bibinfo {volume} {129}},\
  \bibinfo {pages} {270501} (\bibinfo {year} {2022})}\BibitemShut {NoStop}%
\bibitem [{\citenamefont {Mart{\'\i}n}\ \emph {et~al.}(2023)\citenamefont
  {Mart{\'\i}n}, \citenamefont {Plekhanov},\ and\ \citenamefont
  {Lubasch}}]{martin2023barren}%
  \BibitemOpen
  \bibfield  {author} {\bibinfo {author} {\bibfnamefont {E.~C.}\ \bibnamefont
  {Mart{\'\i}n}}, \bibinfo {author} {\bibfnamefont {K.}~\bibnamefont
  {Plekhanov}},\ and\ \bibinfo {author} {\bibfnamefont {M.}~\bibnamefont
  {Lubasch}},\ }\bibfield  {title} {\bibinfo {title} {Barren plateaus in
  quantum tensor network optimization},\ }\href@noop {} {\bibfield  {journal}
  {\bibinfo  {journal} {Quantum}\ }\textbf {\bibinfo {volume} {7}},\ \bibinfo
  {pages} {974} (\bibinfo {year} {2023})}\BibitemShut {NoStop}%
\bibitem [{\citenamefont {Barthel}\ and\ \citenamefont
  {Miao}(2023)}]{barthel2023absence}%
  \BibitemOpen
  \bibfield  {author} {\bibinfo {author} {\bibfnamefont {T.}~\bibnamefont
  {Barthel}}\ and\ \bibinfo {author} {\bibfnamefont {Q.}~\bibnamefont {Miao}},\
  }\bibfield  {title} {\bibinfo {title} {Absence of barren plateaus and scaling
  of gradients in the energy optimization of isometric tensor network states},\
  }\href@noop {} {\bibfield  {journal} {\bibinfo  {journal} {arXiv preprint
  arXiv:2304.00161}\ } (\bibinfo {year} {2023})}\BibitemShut {NoStop}%
\bibitem [{\citenamefont {Zhang}\ \emph {et~al.}(2023)\citenamefont {Zhang},
  \citenamefont {Liu},\ and\ \citenamefont {Zhang}}]{zhang2023absence}%
  \BibitemOpen
  \bibfield  {author} {\bibinfo {author} {\bibfnamefont {H.-K.}\ \bibnamefont
  {Zhang}}, \bibinfo {author} {\bibfnamefont {S.}~\bibnamefont {Liu}},\ and\
  \bibinfo {author} {\bibfnamefont {S.-X.}\ \bibnamefont {Zhang}},\ }\bibfield
  {title} {\bibinfo {title} {Absence of barren plateaus in finite local-depth
  circuits with long-range entanglement},\ }\href@noop {} {\bibfield  {journal}
  {\bibinfo  {journal} {arXiv preprint arXiv:2311.01393}\ } (\bibinfo {year}
  {2023})}\BibitemShut {NoStop}%
\bibitem [{\citenamefont {Cerezo}\ \emph {et~al.}(2023)\citenamefont {Cerezo},
  \citenamefont {Larocca}, \citenamefont {Garc{\'\i}a-Mart{\'\i}n},
  \citenamefont {Diaz}, \citenamefont {Braccia}, \citenamefont {Fontana},
  \citenamefont {Rudolph}, \citenamefont {Bermejo}, \citenamefont {Ijaz},
  \citenamefont {Thanasilp} \emph {et~al.}}]{cerezo2023does}%
  \BibitemOpen
  \bibfield  {author} {\bibinfo {author} {\bibfnamefont {M.}~\bibnamefont
  {Cerezo}}, \bibinfo {author} {\bibfnamefont {M.}~\bibnamefont {Larocca}},
  \bibinfo {author} {\bibfnamefont {D.}~\bibnamefont
  {Garc{\'\i}a-Mart{\'\i}n}}, \bibinfo {author} {\bibfnamefont
  {N.}~\bibnamefont {Diaz}}, \bibinfo {author} {\bibfnamefont {P.}~\bibnamefont
  {Braccia}}, \bibinfo {author} {\bibfnamefont {E.}~\bibnamefont {Fontana}},
  \bibinfo {author} {\bibfnamefont {M.~S.}\ \bibnamefont {Rudolph}}, \bibinfo
  {author} {\bibfnamefont {P.}~\bibnamefont {Bermejo}}, \bibinfo {author}
  {\bibfnamefont {A.}~\bibnamefont {Ijaz}}, \bibinfo {author} {\bibfnamefont
  {S.}~\bibnamefont {Thanasilp}}, \emph {et~al.},\ }\bibfield  {title}
  {\bibinfo {title} {Does provable absence of barren plateaus imply classical
  simulability? or, why we need to rethink variational quantum computing},\
  }\href@noop {} {\bibfield  {journal} {\bibinfo  {journal} {arXiv preprint
  arXiv:2312.09121}\ } (\bibinfo {year} {2023})}\BibitemShut {NoStop}%
\bibitem [{\citenamefont {Skolik}\ \emph {et~al.}(2023)\citenamefont {Skolik},
  \citenamefont {Cattelan}, \citenamefont {Yarkoni}, \citenamefont {B{\"a}ck},\
  and\ \citenamefont {Dunjko}}]{skolik2023equivariant}%
  \BibitemOpen
  \bibfield  {author} {\bibinfo {author} {\bibfnamefont {A.}~\bibnamefont
  {Skolik}}, \bibinfo {author} {\bibfnamefont {M.}~\bibnamefont {Cattelan}},
  \bibinfo {author} {\bibfnamefont {S.}~\bibnamefont {Yarkoni}}, \bibinfo
  {author} {\bibfnamefont {T.}~\bibnamefont {B{\"a}ck}},\ and\ \bibinfo
  {author} {\bibfnamefont {V.}~\bibnamefont {Dunjko}},\ }\bibfield  {title}
  {\bibinfo {title} {Equivariant quantum circuits for learning on weighted
  graphs},\ }\href@noop {} {\bibfield  {journal} {\bibinfo  {journal} {npj
  Quantum Information}\ }\textbf {\bibinfo {volume} {9}},\ \bibinfo {pages}
  {47} (\bibinfo {year} {2023})}\BibitemShut {NoStop}%
\bibitem [{\citenamefont {East}\ \emph {et~al.}(2023)\citenamefont {East},
  \citenamefont {Alonso-Linaje},\ and\ \citenamefont {Park}}]{east2023all}%
  \BibitemOpen
  \bibfield  {author} {\bibinfo {author} {\bibfnamefont {R.~D.}\ \bibnamefont
  {East}}, \bibinfo {author} {\bibfnamefont {G.}~\bibnamefont
  {Alonso-Linaje}},\ and\ \bibinfo {author} {\bibfnamefont {C.-Y.}\
  \bibnamefont {Park}},\ }\bibfield  {title} {\bibinfo {title} {All you need is
  spin: Su (2) equivariant variational quantum circuits based on spin
  networks},\ }\href@noop {} {\bibfield  {journal} {\bibinfo  {journal} {arXiv
  preprint arXiv:2309.07250}\ } (\bibinfo {year} {2023})}\BibitemShut {NoStop}%
\bibitem [{\citenamefont {Schatzki}\ \emph {et~al.}(2023)\citenamefont
  {Schatzki}, \citenamefont {Larocca}, \citenamefont {Sauvage},\ and\
  \citenamefont {Cerezo}}]{schatzki2023theoretical}%
  \BibitemOpen
  \bibfield  {author} {\bibinfo {author} {\bibfnamefont {L.}~\bibnamefont
  {Schatzki}}, \bibinfo {author} {\bibfnamefont {M.}~\bibnamefont {Larocca}},
  \bibinfo {author} {\bibfnamefont {F.}~\bibnamefont {Sauvage}},\ and\ \bibinfo
  {author} {\bibfnamefont {M.}~\bibnamefont {Cerezo}},\ }\bibfield  {title}
  {\bibinfo {title} {Theoretical guarantees for permutation-equivariant quantum
  neural networks},\ }\href@noop {} {\bibfield  {journal} {\bibinfo  {journal}
  {npj Quantum Information}\ }\textbf {\bibinfo {volume} {10}},\ \bibinfo
  {pages} {12} (\bibinfo {year} {2023})}\BibitemShut {NoStop}%
\bibitem [{\citenamefont {Glorot}\ and\ \citenamefont
  {Bengio}(2010)}]{glorot2010understanding}%
  \BibitemOpen
  \bibfield  {author} {\bibinfo {author} {\bibfnamefont {X.}~\bibnamefont
  {Glorot}}\ and\ \bibinfo {author} {\bibfnamefont {Y.}~\bibnamefont
  {Bengio}},\ }\bibfield  {title} {\bibinfo {title} {Understanding the
  difficulty of training deep feedforward neural networks},\ }in\ \href
  {https://proceedings.mlr.press/v9/glorot10a.html} {\emph {\bibinfo
  {booktitle} {Proceedings of the thirteenth international conference on
  artificial intelligence and statistics}}}\ (\bibinfo {organization} {JMLR
  Workshop and Conference Proceedings},\ \bibinfo {year} {2010})\ pp.\ \bibinfo
  {pages} {249--256}\BibitemShut {NoStop}%
\bibitem [{\citenamefont {He}\ \emph {et~al.}(2015)\citenamefont {He},
  \citenamefont {Zhang}, \citenamefont {Ren},\ and\ \citenamefont
  {Sun}}]{he2015delving}%
  \BibitemOpen
  \bibfield  {author} {\bibinfo {author} {\bibfnamefont {K.}~\bibnamefont
  {He}}, \bibinfo {author} {\bibfnamefont {X.}~\bibnamefont {Zhang}}, \bibinfo
  {author} {\bibfnamefont {S.}~\bibnamefont {Ren}},\ and\ \bibinfo {author}
  {\bibfnamefont {J.}~\bibnamefont {Sun}},\ }\bibfield  {title} {\bibinfo
  {title} {Delving deep into rectifiers: Surpassing human-level performance on
  imagenet classification},\ }in\ \href@noop {} {\emph {\bibinfo {booktitle}
  {Proceedings of the IEEE international conference on computer vision}}}\
  (\bibinfo {year} {2015})\ pp.\ \bibinfo {pages} {1026--1034}\BibitemShut
  {NoStop}%
\bibitem [{\citenamefont {Grant}\ \emph {et~al.}(2019)\citenamefont {Grant},
  \citenamefont {Wossnig}, \citenamefont {Ostaszewski},\ and\ \citenamefont
  {Benedetti}}]{grant2019initialization}%
  \BibitemOpen
  \bibfield  {author} {\bibinfo {author} {\bibfnamefont {E.}~\bibnamefont
  {Grant}}, \bibinfo {author} {\bibfnamefont {L.}~\bibnamefont {Wossnig}},
  \bibinfo {author} {\bibfnamefont {M.}~\bibnamefont {Ostaszewski}},\ and\
  \bibinfo {author} {\bibfnamefont {M.}~\bibnamefont {Benedetti}},\ }\bibfield
  {title} {\bibinfo {title} {An initialization strategy for addressing barren
  plateaus in parametrized quantum circuits},\ }\href@noop {} {\bibfield
  {journal} {\bibinfo  {journal} {Quantum}\ }\textbf {\bibinfo {volume} {3}},\
  \bibinfo {pages} {214} (\bibinfo {year} {2019})}\BibitemShut {NoStop}%
\bibitem [{\citenamefont {Jain}\ \emph {et~al.}(2022)\citenamefont {Jain},
  \citenamefont {Coyle}, \citenamefont {Kashefi},\ and\ \citenamefont
  {Kumar}}]{jain2021graph}%
  \BibitemOpen
  \bibfield  {author} {\bibinfo {author} {\bibfnamefont {N.}~\bibnamefont
  {Jain}}, \bibinfo {author} {\bibfnamefont {B.}~\bibnamefont {Coyle}},
  \bibinfo {author} {\bibfnamefont {E.}~\bibnamefont {Kashefi}},\ and\ \bibinfo
  {author} {\bibfnamefont {N.}~\bibnamefont {Kumar}},\ }\bibfield  {title}
  {\bibinfo {title} {Graph neural network initialisation of quantum approximate
  optimisation},\ }\href@noop {} {\bibfield  {journal} {\bibinfo  {journal}
  {{Quantum}}\ }\textbf {\bibinfo {volume} {6}},\ \bibinfo {pages} {861}
  (\bibinfo {year} {2022})}\BibitemShut {NoStop}%
\bibitem [{\citenamefont {Mele}\ \emph {et~al.}(2022)\citenamefont {Mele},
  \citenamefont {Mbeng}, \citenamefont {Santoro}, \citenamefont {Collura},\
  and\ \citenamefont {Torta}}]{mele2022avoiding}%
  \BibitemOpen
  \bibfield  {author} {\bibinfo {author} {\bibfnamefont {A.~A.}\ \bibnamefont
  {Mele}}, \bibinfo {author} {\bibfnamefont {G.~B.}\ \bibnamefont {Mbeng}},
  \bibinfo {author} {\bibfnamefont {G.~E.}\ \bibnamefont {Santoro}}, \bibinfo
  {author} {\bibfnamefont {M.}~\bibnamefont {Collura}},\ and\ \bibinfo {author}
  {\bibfnamefont {P.}~\bibnamefont {Torta}},\ }\bibfield  {title} {\bibinfo
  {title} {Avoiding barren plateaus via transferability of smooth solutions in
  a {H}amiltonian variational ansatz},\ }\href
  {https://link.aps.org/doi/10.1103/PhysRevA.106.L060401} {\bibfield  {journal}
  {\bibinfo  {journal} {Phys. Rev. A}\ }\textbf {\bibinfo {volume} {106}},\
  \bibinfo {pages} {L060401} (\bibinfo {year} {2022})}\BibitemShut {NoStop}%
\bibitem [{\citenamefont {Sack}\ \emph {et~al.}(2022)\citenamefont {Sack},
  \citenamefont {Medina}, \citenamefont {Michailidis}, \citenamefont {Kueng},\
  and\ \citenamefont {Serbyn}}]{sack2022avoiding}%
  \BibitemOpen
  \bibfield  {author} {\bibinfo {author} {\bibfnamefont {S.~H.}\ \bibnamefont
  {Sack}}, \bibinfo {author} {\bibfnamefont {R.~A.}\ \bibnamefont {Medina}},
  \bibinfo {author} {\bibfnamefont {A.~A.}\ \bibnamefont {Michailidis}},
  \bibinfo {author} {\bibfnamefont {R.}~\bibnamefont {Kueng}},\ and\ \bibinfo
  {author} {\bibfnamefont {M.}~\bibnamefont {Serbyn}},\ }\bibfield  {title}
  {\bibinfo {title} {Avoiding barren plateaus using classical shadows},\
  }\href@noop {} {\bibfield  {journal} {\bibinfo  {journal} {PRX Quantum}\
  }\textbf {\bibinfo {volume} {3}},\ \bibinfo {pages} {020365} (\bibinfo {year}
  {2022})}\BibitemShut {NoStop}%
\bibitem [{\citenamefont {Zhang}\ \emph {et~al.}(2022)\citenamefont {Zhang},
  \citenamefont {Liu}, \citenamefont {Hsieh},\ and\ \citenamefont
  {Tao}}]{zhang2022gaussian}%
  \BibitemOpen
  \bibfield  {author} {\bibinfo {author} {\bibfnamefont {K.}~\bibnamefont
  {Zhang}}, \bibinfo {author} {\bibfnamefont {L.}~\bibnamefont {Liu}}, \bibinfo
  {author} {\bibfnamefont {M.-H.}\ \bibnamefont {Hsieh}},\ and\ \bibinfo
  {author} {\bibfnamefont {D.}~\bibnamefont {Tao}},\ }\bibfield  {title}
  {\bibinfo {title} {Escaping from the barren plateau via gaussian
  initializations in deep variational quantum circuits},\ }in\ \href@noop {}
  {\emph {\bibinfo {booktitle} {Advances in Neural Information Processing
  Systems}}},\ Vol.~\bibinfo {volume} {35}\ (\bibinfo {year} {2022})\ pp.\
  \bibinfo {pages} {18612--18627}\BibitemShut {NoStop}%
\bibitem [{\citenamefont {Wang}\ \emph {et~al.}(2023)\citenamefont {Wang},
  \citenamefont {Qi}, \citenamefont {Ferrie},\ and\ \citenamefont
  {Dong}}]{wang2023trainability}%
  \BibitemOpen
  \bibfield  {author} {\bibinfo {author} {\bibfnamefont {Y.}~\bibnamefont
  {Wang}}, \bibinfo {author} {\bibfnamefont {B.}~\bibnamefont {Qi}}, \bibinfo
  {author} {\bibfnamefont {C.}~\bibnamefont {Ferrie}},\ and\ \bibinfo {author}
  {\bibfnamefont {D.}~\bibnamefont {Dong}},\ }\bibfield  {title} {\bibinfo
  {title} {Trainability enhancement of parameterized quantum circuits via
  reduced-domain parameter initialization},\ }\href@noop {} {\bibfield
  {journal} {\bibinfo  {journal} {arXiv preprint arXiv:2302.06858}\ } (\bibinfo
  {year} {2023})}\BibitemShut {NoStop}%
\bibitem [{\citenamefont {Rudolph}\ \emph {et~al.}(2023)\citenamefont
  {Rudolph}, \citenamefont {Miller}, \citenamefont {Motlagh}, \citenamefont
  {Chen}, \citenamefont {Acharya},\ and\ \citenamefont
  {Perdomo-Ortiz}}]{rudolph2023synergistic}%
  \BibitemOpen
  \bibfield  {author} {\bibinfo {author} {\bibfnamefont {M.~S.}\ \bibnamefont
  {Rudolph}}, \bibinfo {author} {\bibfnamefont {J.}~\bibnamefont {Miller}},
  \bibinfo {author} {\bibfnamefont {D.}~\bibnamefont {Motlagh}}, \bibinfo
  {author} {\bibfnamefont {J.}~\bibnamefont {Chen}}, \bibinfo {author}
  {\bibfnamefont {A.}~\bibnamefont {Acharya}},\ and\ \bibinfo {author}
  {\bibfnamefont {A.}~\bibnamefont {Perdomo-Ortiz}},\ }\bibfield  {title}
  {\bibinfo {title} {Synergistic pretraining of parametrized quantum circuits
  via tensor networks},\ }\href@noop {} {\bibfield  {journal} {\bibinfo
  {journal} {Nature Communications}\ }\textbf {\bibinfo {volume} {14}},\
  \bibinfo {pages} {8367} (\bibinfo {year} {2023})}\BibitemShut {NoStop}%
\bibitem [{\citenamefont {Park}\ and\ \citenamefont
  {Killoran}(2024)}]{park2024hamiltonian}%
  \BibitemOpen
  \bibfield  {author} {\bibinfo {author} {\bibfnamefont {C.-Y.}\ \bibnamefont
  {Park}}\ and\ \bibinfo {author} {\bibfnamefont {N.}~\bibnamefont
  {Killoran}},\ }\bibfield  {title} {\bibinfo {title} {Hamiltonian variational
  ansatz without barren plateaus},\ }\href
  {https://doi.org/10.22331/q-2024-02-01-1239} {\bibfield  {journal} {\bibinfo
  {journal} {{Quantum}}\ }\textbf {\bibinfo {volume} {8}},\ \bibinfo {pages}
  {1239} (\bibinfo {year} {2024})}\BibitemShut {NoStop}%
\bibitem [{\citenamefont {Peruzzo}\ \emph {et~al.}(2014)\citenamefont
  {Peruzzo}, \citenamefont {McClean}, \citenamefont {Shadbolt}, \citenamefont
  {Yung}, \citenamefont {Zhou}, \citenamefont {Love}, \citenamefont
  {Aspuru-Guzik},\ and\ \citenamefont {O'{B}rien}}]{peruzzo2014variational}%
  \BibitemOpen
  \bibfield  {author} {\bibinfo {author} {\bibfnamefont {A.}~\bibnamefont
  {Peruzzo}}, \bibinfo {author} {\bibfnamefont {J.}~\bibnamefont {McClean}},
  \bibinfo {author} {\bibfnamefont {P.}~\bibnamefont {Shadbolt}}, \bibinfo
  {author} {\bibfnamefont {M.-H.}\ \bibnamefont {Yung}}, \bibinfo {author}
  {\bibfnamefont {X.-Q.}\ \bibnamefont {Zhou}}, \bibinfo {author}
  {\bibfnamefont {P.~J.}\ \bibnamefont {Love}}, \bibinfo {author}
  {\bibfnamefont {A.}~\bibnamefont {Aspuru-Guzik}},\ and\ \bibinfo {author}
  {\bibfnamefont {J.~L.}\ \bibnamefont {O'{B}rien}},\ }\bibfield  {title}
  {\bibinfo {title} {A variational eigenvalue solver on a photonic quantum
  processor},\ }\href@noop {} {\bibfield  {journal} {\bibinfo  {journal} {Nat.
  Commun.}\ }\textbf {\bibinfo {volume} {5}},\ \bibinfo {pages} {1} (\bibinfo
  {year} {2014})}\BibitemShut {NoStop}%
\bibitem [{\citenamefont {Preskill}(2018)}]{preskill2018quantum}%
  \BibitemOpen
  \bibfield  {author} {\bibinfo {author} {\bibfnamefont {J.}~\bibnamefont
  {Preskill}},\ }\bibfield  {title} {\bibinfo {title} {Quantum computing in the
  {NISQ} era and beyond},\ }\href@noop {} {\bibfield  {journal} {\bibinfo
  {journal} {Quantum}\ }\textbf {\bibinfo {volume} {2}},\ \bibinfo {pages} {79}
  (\bibinfo {year} {2018})}\BibitemShut {NoStop}%
\bibitem [{\citenamefont {Fujii}(2015)}]{fujii2015quantum}%
  \BibitemOpen
  \bibfield  {author} {\bibinfo {author} {\bibfnamefont {K.}~\bibnamefont
  {Fujii}},\ }\href@noop {} {\emph {\bibinfo {title} {Quantum Computation with
  Topological Codes: from qubit to topological fault-tolerance}}},\
  Vol.~\bibinfo {volume} {8}\ (\bibinfo  {publisher} {Springer},\ \bibinfo
  {year} {2015})\BibitemShut {NoStop}%
\bibitem [{\citenamefont {Wiersema}\ \emph {et~al.}(2020)\citenamefont
  {Wiersema}, \citenamefont {Zhou}, \citenamefont {de~Sereville}, \citenamefont
  {Carrasquilla}, \citenamefont {Kim},\ and\ \citenamefont
  {Yuen}}]{wiersema2020exploring}%
  \BibitemOpen
  \bibfield  {author} {\bibinfo {author} {\bibfnamefont {R.}~\bibnamefont
  {Wiersema}}, \bibinfo {author} {\bibfnamefont {C.}~\bibnamefont {Zhou}},
  \bibinfo {author} {\bibfnamefont {Y.}~\bibnamefont {de~Sereville}}, \bibinfo
  {author} {\bibfnamefont {J.~F.}\ \bibnamefont {Carrasquilla}}, \bibinfo
  {author} {\bibfnamefont {Y.~B.}\ \bibnamefont {Kim}},\ and\ \bibinfo {author}
  {\bibfnamefont {H.}~\bibnamefont {Yuen}},\ }\bibfield  {title} {\bibinfo
  {title} {Exploring entanglement and optimization within the {H}amiltonian
  variational ansatz},\ }\href@noop {} {\bibfield  {journal} {\bibinfo
  {journal} {PRX Quantum}\ }\textbf {\bibinfo {volume} {1}},\ \bibinfo {pages}
  {020319} (\bibinfo {year} {2020})}\BibitemShut {NoStop}%
\bibitem [{\citenamefont {Cerezo}\ \emph
  {et~al.}(2021{\natexlab{b}})\citenamefont {Cerezo}, \citenamefont {Sone},
  \citenamefont {Volkoff}, \citenamefont {Cincio},\ and\ \citenamefont
  {Coles}}]{cerezo2021cost}%
  \BibitemOpen
  \bibfield  {author} {\bibinfo {author} {\bibfnamefont {M.}~\bibnamefont
  {Cerezo}}, \bibinfo {author} {\bibfnamefont {A.}~\bibnamefont {Sone}},
  \bibinfo {author} {\bibfnamefont {T.}~\bibnamefont {Volkoff}}, \bibinfo
  {author} {\bibfnamefont {L.}~\bibnamefont {Cincio}},\ and\ \bibinfo {author}
  {\bibfnamefont {P.~J.}\ \bibnamefont {Coles}},\ }\bibfield  {title} {\bibinfo
  {title} {Cost function dependent barren plateaus in shallow parametrized
  quantum circuits},\ }\href@noop {} {\bibfield  {journal} {\bibinfo  {journal}
  {Nat. Commun.}\ }\textbf {\bibinfo {volume} {12}},\ \bibinfo {pages} {1}
  (\bibinfo {year} {2021}{\natexlab{b}})}\BibitemShut {NoStop}%
\bibitem [{\citenamefont {Shtanko}\ \emph {et~al.}(2023)\citenamefont
  {Shtanko}, \citenamefont {Wang}, \citenamefont {Zhang}, \citenamefont
  {Harle}, \citenamefont {Seif}, \citenamefont {Movassagh},\ and\ \citenamefont
  {Minev}}]{shtanko2023uncovering}%
  \BibitemOpen
  \bibfield  {author} {\bibinfo {author} {\bibfnamefont {O.}~\bibnamefont
  {Shtanko}}, \bibinfo {author} {\bibfnamefont {D.~S.}\ \bibnamefont {Wang}},
  \bibinfo {author} {\bibfnamefont {H.}~\bibnamefont {Zhang}}, \bibinfo
  {author} {\bibfnamefont {N.}~\bibnamefont {Harle}}, \bibinfo {author}
  {\bibfnamefont {A.}~\bibnamefont {Seif}}, \bibinfo {author} {\bibfnamefont
  {R.}~\bibnamefont {Movassagh}},\ and\ \bibinfo {author} {\bibfnamefont
  {Z.}~\bibnamefont {Minev}},\ }\bibfield  {title} {\bibinfo {title}
  {Uncovering local integrability in quantum many-body dynamics},\ }\href@noop
  {} {\bibfield  {journal} {\bibinfo  {journal} {arXiv preprint
  arXiv:2307.07552}\ } (\bibinfo {year} {2023})}\BibitemShut {NoStop}%
\bibitem [{\citenamefont {Huse}\ \emph {et~al.}(2014)\citenamefont {Huse},
  \citenamefont {Nandkishore},\ and\ \citenamefont
  {Oganesyan}}]{huse2014phenomenology}%
  \BibitemOpen
  \bibfield  {author} {\bibinfo {author} {\bibfnamefont {D.~A.}\ \bibnamefont
  {Huse}}, \bibinfo {author} {\bibfnamefont {R.}~\bibnamefont {Nandkishore}},\
  and\ \bibinfo {author} {\bibfnamefont {V.}~\bibnamefont {Oganesyan}},\
  }\bibfield  {title} {\bibinfo {title} {Phenomenology of fully
  many-body-localized systems},\ }\href@noop {} {\bibfield  {journal} {\bibinfo
   {journal} {Phys. Rev. B}\ }\textbf {\bibinfo {volume} {90}},\ \bibinfo
  {pages} {174202} (\bibinfo {year} {2014})}\BibitemShut {NoStop}%
\bibitem [{\citenamefont {Cirac}\ and\ \citenamefont
  {Zoller}(1995)}]{cirac1995quantum}%
  \BibitemOpen
  \bibfield  {author} {\bibinfo {author} {\bibfnamefont {J.~I.}\ \bibnamefont
  {Cirac}}\ and\ \bibinfo {author} {\bibfnamefont {P.}~\bibnamefont {Zoller}},\
  }\bibfield  {title} {\bibinfo {title} {Quantum computations with cold trapped
  ions},\ }\href@noop {} {\bibfield  {journal} {\bibinfo  {journal} {Phys. Rev.
  Lett.}\ }\textbf {\bibinfo {volume} {74}},\ \bibinfo {pages} {4091} (\bibinfo
  {year} {1995})}\BibitemShut {NoStop}%
\bibitem [{\citenamefont {Jaksch}\ \emph {et~al.}(2000)\citenamefont {Jaksch},
  \citenamefont {Cirac}, \citenamefont {Zoller}, \citenamefont {Rolston},
  \citenamefont {C{\^o}t{\'e}},\ and\ \citenamefont {Lukin}}]{jaksch2000fast}%
  \BibitemOpen
  \bibfield  {author} {\bibinfo {author} {\bibfnamefont {D.}~\bibnamefont
  {Jaksch}}, \bibinfo {author} {\bibfnamefont {J.~I.}\ \bibnamefont {Cirac}},
  \bibinfo {author} {\bibfnamefont {P.}~\bibnamefont {Zoller}}, \bibinfo
  {author} {\bibfnamefont {S.~L.}\ \bibnamefont {Rolston}}, \bibinfo {author}
  {\bibfnamefont {R.}~\bibnamefont {C{\^o}t{\'e}}},\ and\ \bibinfo {author}
  {\bibfnamefont {M.~D.}\ \bibnamefont {Lukin}},\ }\bibfield  {title} {\bibinfo
  {title} {Fast quantum gates for neutral atoms},\ }\href@noop {} {\bibfield
  {journal} {\bibinfo  {journal} {Phys. Rev. Lett.}\ }\textbf {\bibinfo
  {volume} {85}},\ \bibinfo {pages} {2208} (\bibinfo {year}
  {2000})}\BibitemShut {NoStop}%
\bibitem [{\citenamefont {Chow}\ \emph {et~al.}(2011)\citenamefont {Chow},
  \citenamefont {C{\'o}rcoles}, \citenamefont {Gambetta}, \citenamefont
  {Rigetti}, \citenamefont {Johnson}, \citenamefont {Smolin}, \citenamefont
  {Rozen}, \citenamefont {Keefe}, \citenamefont {Rothwell}, \citenamefont
  {Ketchen} \emph {et~al.}}]{chow2011simple}%
  \BibitemOpen
  \bibfield  {author} {\bibinfo {author} {\bibfnamefont {J.~M.}\ \bibnamefont
  {Chow}}, \bibinfo {author} {\bibfnamefont {A.~D.}\ \bibnamefont
  {C{\'o}rcoles}}, \bibinfo {author} {\bibfnamefont {J.~M.}\ \bibnamefont
  {Gambetta}}, \bibinfo {author} {\bibfnamefont {C.}~\bibnamefont {Rigetti}},
  \bibinfo {author} {\bibfnamefont {B.~R.}\ \bibnamefont {Johnson}}, \bibinfo
  {author} {\bibfnamefont {J.~A.}\ \bibnamefont {Smolin}}, \bibinfo {author}
  {\bibfnamefont {J.~R.}\ \bibnamefont {Rozen}}, \bibinfo {author}
  {\bibfnamefont {G.~A.}\ \bibnamefont {Keefe}}, \bibinfo {author}
  {\bibfnamefont {M.~B.}\ \bibnamefont {Rothwell}}, \bibinfo {author}
  {\bibfnamefont {M.~B.}\ \bibnamefont {Ketchen}}, \emph {et~al.},\ }\bibfield
  {title} {\bibinfo {title} {Simple all-microwave entangling gate for
  fixed-frequency superconducting qubits},\ }\href@noop {} {\bibfield
  {journal} {\bibinfo  {journal} {Phys. Rev. Lett.}\ }\textbf {\bibinfo
  {volume} {107}},\ \bibinfo {pages} {080502} (\bibinfo {year}
  {2011})}\BibitemShut {NoStop}%
\bibitem [{\citenamefont {Figgatt}\ \emph {et~al.}(2019)\citenamefont
  {Figgatt}, \citenamefont {Ostrander}, \citenamefont {Linke}, \citenamefont
  {Landsman}, \citenamefont {Zhu}, \citenamefont {Maslov},\ and\ \citenamefont
  {Monroe}}]{figgatt2019parallel}%
  \BibitemOpen
  \bibfield  {author} {\bibinfo {author} {\bibfnamefont {C.}~\bibnamefont
  {Figgatt}}, \bibinfo {author} {\bibfnamefont {A.}~\bibnamefont {Ostrander}},
  \bibinfo {author} {\bibfnamefont {N.~M.}\ \bibnamefont {Linke}}, \bibinfo
  {author} {\bibfnamefont {K.~A.}\ \bibnamefont {Landsman}}, \bibinfo {author}
  {\bibfnamefont {D.}~\bibnamefont {Zhu}}, \bibinfo {author} {\bibfnamefont
  {D.}~\bibnamefont {Maslov}},\ and\ \bibinfo {author} {\bibfnamefont
  {C.}~\bibnamefont {Monroe}},\ }\bibfield  {title} {\bibinfo {title} {Parallel
  entangling operations on a universal ion-trap quantum computer},\ }\href@noop
  {} {\bibfield  {journal} {\bibinfo  {journal} {Nature}\ }\textbf {\bibinfo
  {volume} {572}},\ \bibinfo {pages} {368} (\bibinfo {year}
  {2019})}\BibitemShut {NoStop}%
\bibitem [{\citenamefont {Kim}\ \emph {et~al.}(2023)\citenamefont {Kim},
  \citenamefont {Eddins}, \citenamefont {Anand}, \citenamefont {Wei},
  \citenamefont {Van Den~Berg}, \citenamefont {Rosenblatt}, \citenamefont
  {Nayfeh}, \citenamefont {Wu}, \citenamefont {Zaletel}, \citenamefont {Temme}
  \emph {et~al.}}]{kim2023evidence}%
  \BibitemOpen
  \bibfield  {author} {\bibinfo {author} {\bibfnamefont {Y.}~\bibnamefont
  {Kim}}, \bibinfo {author} {\bibfnamefont {A.}~\bibnamefont {Eddins}},
  \bibinfo {author} {\bibfnamefont {S.}~\bibnamefont {Anand}}, \bibinfo
  {author} {\bibfnamefont {K.~X.}\ \bibnamefont {Wei}}, \bibinfo {author}
  {\bibfnamefont {E.}~\bibnamefont {Van Den~Berg}}, \bibinfo {author}
  {\bibfnamefont {S.}~\bibnamefont {Rosenblatt}}, \bibinfo {author}
  {\bibfnamefont {H.}~\bibnamefont {Nayfeh}}, \bibinfo {author} {\bibfnamefont
  {Y.}~\bibnamefont {Wu}}, \bibinfo {author} {\bibfnamefont {M.}~\bibnamefont
  {Zaletel}}, \bibinfo {author} {\bibfnamefont {K.}~\bibnamefont {Temme}},
  \emph {et~al.},\ }\bibfield  {title} {\bibinfo {title} {Evidence for the
  utility of quantum computing before fault tolerance},\ }\href@noop {}
  {\bibfield  {journal} {\bibinfo  {journal} {Nature}\ }\textbf {\bibinfo
  {volume} {618}},\ \bibinfo {pages} {500} (\bibinfo {year}
  {2023})}\BibitemShut {NoStop}%
\bibitem [{\citenamefont {Evered}\ \emph {et~al.}(2023)\citenamefont {Evered},
  \citenamefont {Bluvstein}, \citenamefont {Kalinowski}, \citenamefont {Ebadi},
  \citenamefont {Manovitz}, \citenamefont {Zhou}, \citenamefont {Li},
  \citenamefont {Geim}, \citenamefont {Wang}, \citenamefont {Maskara},
  \citenamefont {Levine}, \citenamefont {Semeghini}, \citenamefont {Greiner},
  \citenamefont {Vuletić},\ and\ \citenamefont {Lukin}}]{evered2023high}%
  \BibitemOpen
  \bibfield  {author} {\bibinfo {author} {\bibfnamefont {S.~J.}\ \bibnamefont
  {Evered}}, \bibinfo {author} {\bibfnamefont {D.}~\bibnamefont {Bluvstein}},
  \bibinfo {author} {\bibfnamefont {M.}~\bibnamefont {Kalinowski}}, \bibinfo
  {author} {\bibfnamefont {S.}~\bibnamefont {Ebadi}}, \bibinfo {author}
  {\bibfnamefont {T.}~\bibnamefont {Manovitz}}, \bibinfo {author}
  {\bibfnamefont {H.}~\bibnamefont {Zhou}}, \bibinfo {author} {\bibfnamefont
  {S.~H.}\ \bibnamefont {Li}}, \bibinfo {author} {\bibfnamefont {A.~A.}\
  \bibnamefont {Geim}}, \bibinfo {author} {\bibfnamefont {T.~T.}\ \bibnamefont
  {Wang}}, \bibinfo {author} {\bibfnamefont {N.}~\bibnamefont {Maskara}},
  \bibinfo {author} {\bibfnamefont {H.}~\bibnamefont {Levine}}, \bibinfo
  {author} {\bibfnamefont {G.}~\bibnamefont {Semeghini}}, \bibinfo {author}
  {\bibfnamefont {M.}~\bibnamefont {Greiner}}, \bibinfo {author} {\bibfnamefont
  {V.}~\bibnamefont {Vuletić}},\ and\ \bibinfo {author} {\bibfnamefont
  {M.~D.}\ \bibnamefont {Lukin}},\ }\bibfield  {title} {\bibinfo {title}
  {High-fidelity parallel entangling gates on a neutral-atom quantum
  computer},\ }\href {https://doi.org/10.1038/s41586-023-06481-y} {\bibfield
  {journal} {\bibinfo  {journal} {Nature}\ }\textbf {\bibinfo {volume} {622}},\
  \bibinfo {pages} {268–272} (\bibinfo {year} {2023})}\BibitemShut {NoStop}%
\bibitem [{hea()}]{heavy-hex-lattice}%
  \BibitemOpen
  \href@noop {} {\bibinfo {title} {The {IBM} {Q}uantum heavy hex lattice}},\
  \bibinfo {howpublished}
  {\url{https://research.ibm.com/blog/heavy-hex-lattice}}\BibitemShut {NoStop}%
\bibitem [{\citenamefont {Wessels}\ and\ \citenamefont
  {Barnard}(1992)}]{wessels1992avoiding}%
  \BibitemOpen
  \bibfield  {author} {\bibinfo {author} {\bibfnamefont {L.~F.}\ \bibnamefont
  {Wessels}}\ and\ \bibinfo {author} {\bibfnamefont {E.}~\bibnamefont
  {Barnard}},\ }\bibfield  {title} {\bibinfo {title} {Avoiding false local
  minima by proper initialization of connections},\ }\href@noop {} {\bibfield
  {journal} {\bibinfo  {journal} {IEEE transactions on neural networks}\
  }\textbf {\bibinfo {volume} {3}},\ \bibinfo {pages} {899} (\bibinfo {year}
  {1992})}\BibitemShut {NoStop}%
\bibitem [{\citenamefont {Serbyn}\ \emph {et~al.}(2014)\citenamefont {Serbyn},
  \citenamefont {Papi{\'c}},\ and\ \citenamefont {Abanin}}]{serbyn2014quantum}%
  \BibitemOpen
  \bibfield  {author} {\bibinfo {author} {\bibfnamefont {M.}~\bibnamefont
  {Serbyn}}, \bibinfo {author} {\bibfnamefont {Z.}~\bibnamefont {Papi{\'c}}},\
  and\ \bibinfo {author} {\bibfnamefont {D.~A.}\ \bibnamefont {Abanin}},\
  }\bibfield  {title} {\bibinfo {title} {Quantum quenches in the many-body
  localized phase},\ }\href@noop {} {\bibfield  {journal} {\bibinfo  {journal}
  {Phys. Rev. B}\ }\textbf {\bibinfo {volume} {90}},\ \bibinfo {pages} {174302}
  (\bibinfo {year} {2014})}\BibitemShut {NoStop}%
\bibitem [{\citenamefont {Thiery}\ \emph {et~al.}(2018)\citenamefont {Thiery},
  \citenamefont {Huveneers}, \citenamefont {M{\"u}ller},\ and\ \citenamefont
  {De~Roeck}}]{thiery2018many}%
  \BibitemOpen
  \bibfield  {author} {\bibinfo {author} {\bibfnamefont {T.}~\bibnamefont
  {Thiery}}, \bibinfo {author} {\bibfnamefont {F.}~\bibnamefont {Huveneers}},
  \bibinfo {author} {\bibfnamefont {M.}~\bibnamefont {M{\"u}ller}},\ and\
  \bibinfo {author} {\bibfnamefont {W.}~\bibnamefont {De~Roeck}},\ }\bibfield
  {title} {\bibinfo {title} {Many-body delocalization as a quantum avalanche},\
  }\href@noop {} {\bibfield  {journal} {\bibinfo  {journal} {Phys. Rev. Lett.}\
  }\textbf {\bibinfo {volume} {121}},\ \bibinfo {pages} {140601} (\bibinfo
  {year} {2018})}\BibitemShut {NoStop}%
\bibitem [{\citenamefont {Wahl}\ \emph {et~al.}(2019)\citenamefont {Wahl},
  \citenamefont {Pal},\ and\ \citenamefont {Simon}}]{wahl2019signatures}%
  \BibitemOpen
  \bibfield  {author} {\bibinfo {author} {\bibfnamefont {T.~B.}\ \bibnamefont
  {Wahl}}, \bibinfo {author} {\bibfnamefont {A.}~\bibnamefont {Pal}},\ and\
  \bibinfo {author} {\bibfnamefont {S.~H.}\ \bibnamefont {Simon}},\ }\bibfield
  {title} {\bibinfo {title} {Signatures of the many-body localized regime in
  two dimensions},\ }\href@noop {} {\bibfield  {journal} {\bibinfo  {journal}
  {Nat. Phys.}\ }\textbf {\bibinfo {volume} {15}},\ \bibinfo {pages} {164}
  (\bibinfo {year} {2019})}\BibitemShut {NoStop}%
\bibitem [{Note1()}]{Note1}%
  \BibitemOpen
  \bibinfo {note} {The value of $\gamma $ obtained from the proof of
  Theorem~\ref {thm:large_gradients_constraints} can be smaller than $\pi $.
  However, since many inequalities we have used for the proof are not tight, we
  regard Theorem~\ref {thm:large_gradients_constraints} as a rough estimation
  of the order of such a parameter condition. The exact scaling of the
  parameter condition can be further investigated numerically. See also
  Appendix~\ref {app:additional_numerical_results}.}\BibitemShut {Stop}%
\bibitem [{\citenamefont {Kingma}\ and\ \citenamefont
  {Ba}(2015)}]{kingma2014adam}%
  \BibitemOpen
  \bibfield  {author} {\bibinfo {author} {\bibfnamefont {D.~P.}\ \bibnamefont
  {Kingma}}\ and\ \bibinfo {author} {\bibfnamefont {J.}~\bibnamefont {Ba}},\
  }\bibfield  {title} {\bibinfo {title} {Adam: {A} method for stochastic
  optimization},\ }in\ \href {https://doi.org/10.48550/arXiv.1412.6980} {\emph
  {\bibinfo {booktitle} {3rd International Conference on Learning
  Representations, {ICLR} 2015, San Diego, CA, USA, May 7-9, 2015, Conference
  Track Proceedings}}}\ (\bibinfo {year} {2015})\BibitemShut {NoStop}%
\bibitem [{\citenamefont {Wecker}\ \emph {et~al.}(2015)\citenamefont {Wecker},
  \citenamefont {Hastings},\ and\ \citenamefont {Troyer}}]{wecker2015progress}%
  \BibitemOpen
  \bibfield  {author} {\bibinfo {author} {\bibfnamefont {D.}~\bibnamefont
  {Wecker}}, \bibinfo {author} {\bibfnamefont {M.~B.}\ \bibnamefont
  {Hastings}},\ and\ \bibinfo {author} {\bibfnamefont {M.}~\bibnamefont
  {Troyer}},\ }\bibfield  {title} {\bibinfo {title} {Progress towards practical
  quantum variational algorithms},\ }\href@noop {} {\bibfield  {journal}
  {\bibinfo  {journal} {Phys. Rev. A}\ }\textbf {\bibinfo {volume} {92}},\
  \bibinfo {pages} {042303} (\bibinfo {year} {2015})}\BibitemShut {NoStop}%
\bibitem [{\citenamefont {Hadfield}\ \emph {et~al.}(2019)\citenamefont
  {Hadfield}, \citenamefont {Wang}, \citenamefont {O’Gorman}, \citenamefont
  {Rieffel}, \citenamefont {Venturelli},\ and\ \citenamefont
  {Biswas}}]{hadfield2019quantum}%
  \BibitemOpen
  \bibfield  {author} {\bibinfo {author} {\bibfnamefont {S.}~\bibnamefont
  {Hadfield}}, \bibinfo {author} {\bibfnamefont {Z.}~\bibnamefont {Wang}},
  \bibinfo {author} {\bibfnamefont {B.}~\bibnamefont {O’Gorman}}, \bibinfo
  {author} {\bibfnamefont {E.~G.}\ \bibnamefont {Rieffel}}, \bibinfo {author}
  {\bibfnamefont {D.}~\bibnamefont {Venturelli}},\ and\ \bibinfo {author}
  {\bibfnamefont {R.}~\bibnamefont {Biswas}},\ }\bibfield  {title} {\bibinfo
  {title} {From the quantum approximate optimization algorithm to a quantum
  alternating operator ansatz},\ }\href@noop {} {\bibfield  {journal} {\bibinfo
   {journal} {Algorithms}\ }\textbf {\bibinfo {volume} {12}},\ \bibinfo {pages}
  {34} (\bibinfo {year} {2019})}\BibitemShut {NoStop}%
\bibitem [{\citenamefont {Park}(2024)}]{park2021efficient}%
  \BibitemOpen
  \bibfield  {author} {\bibinfo {author} {\bibfnamefont {C.-Y.}\ \bibnamefont
  {Park}},\ }\bibfield  {title} {\bibinfo {title} {Efficient ground state
  preparation in variational quantum eigensolver with symmetry breaking
  layers},\ }\href@noop {} {\bibfield  {journal} {\bibinfo  {journal} {APL
  Quantum}\ }\textbf {\bibinfo {volume} {1}},\ \bibinfo {pages} {016101}
  (\bibinfo {year} {2024})}\BibitemShut {NoStop}%
\bibitem [{Note2()}]{Note2}%
  \BibitemOpen
  \bibinfo {note} {In contrast, the HEA within the first parameter condition is
  classically simulable up to an inverse polynomial additive error using
  algorithms developed in Refs.~\cite {bravyi2021classical,coble2022quasi}.
  However, we expect that one can extend Theorem~\ref
  {thm:large_gradients_constraints} to an arbitrary graph (instead of a
  $D$-dimensional lattice), which makes those algorithms cannot be
  applied.}\BibitemShut {Stop}%
\bibitem [{\citenamefont {Marrero}\ \emph {et~al.}(2021)\citenamefont
  {Marrero}, \citenamefont {Kieferov{\'a}},\ and\ \citenamefont
  {Wiebe}}]{marrero2021entanglement}%
  \BibitemOpen
  \bibfield  {author} {\bibinfo {author} {\bibfnamefont {C.~O.}\ \bibnamefont
  {Marrero}}, \bibinfo {author} {\bibfnamefont {M.}~\bibnamefont
  {Kieferov{\'a}}},\ and\ \bibinfo {author} {\bibfnamefont {N.}~\bibnamefont
  {Wiebe}},\ }\bibfield  {title} {\bibinfo {title} {Entanglement-induced barren
  plateaus},\ }\href@noop {} {\bibfield  {journal} {\bibinfo  {journal} {PRX
  Quantum}\ }\textbf {\bibinfo {volume} {2}},\ \bibinfo {pages} {040316}
  (\bibinfo {year} {2021})}\BibitemShut {NoStop}%
\bibitem [{\citenamefont {Leone}\ \emph {et~al.}(2022)\citenamefont {Leone},
  \citenamefont {Oliviero}, \citenamefont {Cincio},\ and\ \citenamefont
  {Cerezo}}]{leone2022practical}%
  \BibitemOpen
  \bibfield  {author} {\bibinfo {author} {\bibfnamefont {L.}~\bibnamefont
  {Leone}}, \bibinfo {author} {\bibfnamefont {S.~F.}\ \bibnamefont {Oliviero}},
  \bibinfo {author} {\bibfnamefont {L.}~\bibnamefont {Cincio}},\ and\ \bibinfo
  {author} {\bibfnamefont {M.}~\bibnamefont {Cerezo}},\ }\bibfield  {title}
  {\bibinfo {title} {On the practical usefulness of the hardware efficient
  ansatz},\ }\href@noop {} {\bibfield  {journal} {\bibinfo  {journal} {arXiv
  preprint arXiv:2211.01477}\ } (\bibinfo {year} {2022})}\BibitemShut {NoStop}%
\bibitem [{\citenamefont {Bardarson}\ \emph {et~al.}(2012)\citenamefont
  {Bardarson}, \citenamefont {Pollmann},\ and\ \citenamefont
  {Moore}}]{bardarson2012unbounded}%
  \BibitemOpen
  \bibfield  {author} {\bibinfo {author} {\bibfnamefont {J.~H.}\ \bibnamefont
  {Bardarson}}, \bibinfo {author} {\bibfnamefont {F.}~\bibnamefont
  {Pollmann}},\ and\ \bibinfo {author} {\bibfnamefont {J.~E.}\ \bibnamefont
  {Moore}},\ }\bibfield  {title} {\bibinfo {title} {Unbounded growth of
  entanglement in models of many-body localization},\ }\href@noop {} {\bibfield
   {journal} {\bibinfo  {journal} {Phys. Rev. Lett.}\ }\textbf {\bibinfo
  {volume} {109}},\ \bibinfo {pages} {017202} (\bibinfo {year}
  {2012})}\BibitemShut {NoStop}%
\bibitem [{\citenamefont {Shi}\ and\ \citenamefont
  {Shang}(2024)}]{shi2024avoiding}%
  \BibitemOpen
  \bibfield  {author} {\bibinfo {author} {\bibfnamefont {X.}~\bibnamefont
  {Shi}}\ and\ \bibinfo {author} {\bibfnamefont {Y.}~\bibnamefont {Shang}},\
  }\bibfield  {title} {\bibinfo {title} {Avoiding barren plateaus via gaussian
  mixture model},\ }\href@noop {} {\bibfield  {journal} {\bibinfo  {journal}
  {arXiv preprint arXiv:2402.13501}\ } (\bibinfo {year} {2024})}\BibitemShut
  {NoStop}%
\bibitem [{\citenamefont {Bergholm}\ \emph {et~al.}(2018)\citenamefont
  {Bergholm}, \citenamefont {Izaac}, \citenamefont {Schuld}, \citenamefont
  {Gogolin}, \citenamefont {Ahmed}, \citenamefont {Ajith}, \citenamefont
  {Alam}, \citenamefont {Alonso-Linaje} \emph
  {et~al.}}]{bergholm2018pennylane}%
  \BibitemOpen
  \bibfield  {author} {\bibinfo {author} {\bibfnamefont {V.}~\bibnamefont
  {Bergholm}}, \bibinfo {author} {\bibfnamefont {J.}~\bibnamefont {Izaac}},
  \bibinfo {author} {\bibfnamefont {M.}~\bibnamefont {Schuld}}, \bibinfo
  {author} {\bibfnamefont {C.}~\bibnamefont {Gogolin}}, \bibinfo {author}
  {\bibfnamefont {S.}~\bibnamefont {Ahmed}}, \bibinfo {author} {\bibfnamefont
  {V.}~\bibnamefont {Ajith}}, \bibinfo {author} {\bibfnamefont {M.~S.}\
  \bibnamefont {Alam}}, \bibinfo {author} {\bibfnamefont {G.}~\bibnamefont
  {Alonso-Linaje}}, \emph {et~al.},\ }\href@noop {} {\bibinfo {title}
  {Pennylane: Automatic differentiation of hybrid quantum-classical
  computations}} (\bibinfo {year} {2018}),\ \Eprint
  {https://arxiv.org/abs/1811.04968} {arXiv:1811.04968 [quant-ph]} \BibitemShut
  {NoStop}%
\bibitem [{\citenamefont {Asadi}\ \emph {et~al.}(2024)\citenamefont {Asadi},
  \citenamefont {Dusko}, \citenamefont {Park}, \citenamefont {Michaud-Rioux},
  \citenamefont {Schoch}, \citenamefont {Shu}, \citenamefont {Vincent},\ and\
  \citenamefont {O'Riordan}}]{Lightning}%
  \BibitemOpen
  \bibfield  {author} {\bibinfo {author} {\bibfnamefont {A.}~\bibnamefont
  {Asadi}}, \bibinfo {author} {\bibfnamefont {A.}~\bibnamefont {Dusko}},
  \bibinfo {author} {\bibfnamefont {C.-Y.}\ \bibnamefont {Park}}, \bibinfo
  {author} {\bibfnamefont {V.}~\bibnamefont {Michaud-Rioux}}, \bibinfo {author}
  {\bibfnamefont {I.}~\bibnamefont {Schoch}}, \bibinfo {author} {\bibfnamefont
  {S.}~\bibnamefont {Shu}}, \bibinfo {author} {\bibfnamefont {T.}~\bibnamefont
  {Vincent}},\ and\ \bibinfo {author} {\bibfnamefont {L.~J.}\ \bibnamefont
  {O'Riordan}},\ }\href@noop {} {\bibinfo {title} {Hybrid quantum programming
  with pennylane lightning on {HPC} platforms}} (\bibinfo {year} {2024}),\
  \Eprint {https://arxiv.org/abs/2403.02512} {arXiv:2403.02512 [quant-ph]}
  \BibitemShut {NoStop}%
\bibitem [{\citenamefont {Bravyi}\ \emph {et~al.}(2021)\citenamefont {Bravyi},
  \citenamefont {Gosset},\ and\ \citenamefont
  {Movassagh}}]{bravyi2021classical}%
  \BibitemOpen
  \bibfield  {author} {\bibinfo {author} {\bibfnamefont {S.}~\bibnamefont
  {Bravyi}}, \bibinfo {author} {\bibfnamefont {D.}~\bibnamefont {Gosset}},\
  and\ \bibinfo {author} {\bibfnamefont {R.}~\bibnamefont {Movassagh}},\
  }\bibfield  {title} {\bibinfo {title} {Classical algorithms for quantum mean
  values},\ }\href@noop {} {\bibfield  {journal} {\bibinfo  {journal} {Nat.
  Phys.}\ }\textbf {\bibinfo {volume} {17}},\ \bibinfo {pages} {337} (\bibinfo
  {year} {2021})}\BibitemShut {NoStop}%
\bibitem [{\citenamefont {Coble}\ and\ \citenamefont
  {Coudron}(2022)}]{coble2022quasi}%
  \BibitemOpen
  \bibfield  {author} {\bibinfo {author} {\bibfnamefont {N.~J.}\ \bibnamefont
  {Coble}}\ and\ \bibinfo {author} {\bibfnamefont {M.}~\bibnamefont
  {Coudron}},\ }\bibfield  {title} {\bibinfo {title} {Quasi-polynomial time
  approximation of output probabilities of geometrically-local, shallow quantum
  circuits},\ }in\ \href@noop {} {\emph {\bibinfo {booktitle} {2021 IEEE 62nd
  Annual Symposium on Foundations of Computer Science (FOCS)}}}\ (\bibinfo
  {organization} {IEEE},\ \bibinfo {year} {2022})\ pp.\ \bibinfo {pages}
  {598--609}\BibitemShut {NoStop}%
\bibitem [{\citenamefont {Kuwahara}\ \emph {et~al.}(2016)\citenamefont
  {Kuwahara}, \citenamefont {Mori},\ and\ \citenamefont
  {Saito}}]{kuwahara2016floquet}%
  \BibitemOpen
  \bibfield  {author} {\bibinfo {author} {\bibfnamefont {T.}~\bibnamefont
  {Kuwahara}}, \bibinfo {author} {\bibfnamefont {T.}~\bibnamefont {Mori}},\
  and\ \bibinfo {author} {\bibfnamefont {K.}~\bibnamefont {Saito}},\ }\bibfield
   {title} {\bibinfo {title} {{F}loquet--{M}agnus theory and generic transient
  dynamics in periodically driven many-body quantum systems},\ }\href@noop {}
  {\bibfield  {journal} {\bibinfo  {journal} {Annals of Physics}\ }\textbf
  {\bibinfo {volume} {367}},\ \bibinfo {pages} {96} (\bibinfo {year}
  {2016})}\BibitemShut {NoStop}%
\bibitem [{\citenamefont {Eisert}\ \emph {et~al.}(2015)\citenamefont {Eisert},
  \citenamefont {Friesdorf},\ and\ \citenamefont
  {Gogolin}}]{eisert2015quantum}%
  \BibitemOpen
  \bibfield  {author} {\bibinfo {author} {\bibfnamefont {J.}~\bibnamefont
  {Eisert}}, \bibinfo {author} {\bibfnamefont {M.}~\bibnamefont {Friesdorf}},\
  and\ \bibinfo {author} {\bibfnamefont {C.}~\bibnamefont {Gogolin}},\
  }\bibfield  {title} {\bibinfo {title} {Quantum many-body systems out of
  equilibrium},\ }\href@noop {} {\bibfield  {journal} {\bibinfo  {journal}
  {Nat. Phys.}\ }\textbf {\bibinfo {volume} {11}},\ \bibinfo {pages} {124}
  (\bibinfo {year} {2015})}\BibitemShut {NoStop}%
\bibitem [{\citenamefont {Gogolin}\ and\ \citenamefont
  {Eisert}(2016)}]{gogolin2016equilibration}%
  \BibitemOpen
  \bibfield  {author} {\bibinfo {author} {\bibfnamefont {C.}~\bibnamefont
  {Gogolin}}\ and\ \bibinfo {author} {\bibfnamefont {J.}~\bibnamefont
  {Eisert}},\ }\bibfield  {title} {\bibinfo {title} {Equilibration,
  thermalisation, and the emergence of statistical mechanics in closed quantum
  systems},\ }\href@noop {} {\bibfield  {journal} {\bibinfo  {journal} {Rep.
  Prog. Phys}\ }\textbf {\bibinfo {volume} {79}},\ \bibinfo {pages} {056001}
  (\bibinfo {year} {2016})}\BibitemShut {NoStop}%
\bibitem [{\citenamefont {Nandkishore}\ and\ \citenamefont
  {Huse}(2015)}]{nandkishore2015many}%
  \BibitemOpen
  \bibfield  {author} {\bibinfo {author} {\bibfnamefont {R.}~\bibnamefont
  {Nandkishore}}\ and\ \bibinfo {author} {\bibfnamefont {D.~A.}\ \bibnamefont
  {Huse}},\ }\bibfield  {title} {\bibinfo {title} {Many-body localization and
  thermalization in quantum statistical mechanics},\ }\href@noop {} {\bibfield
  {journal} {\bibinfo  {journal} {Annu. Rev. Condens. Matter Phys.}\ }\textbf
  {\bibinfo {volume} {6}},\ \bibinfo {pages} {15} (\bibinfo {year}
  {2015})}\BibitemShut {NoStop}%
\bibitem [{\citenamefont {D'Alessio}\ \emph {et~al.}(2016)\citenamefont
  {D'Alessio}, \citenamefont {Kafri}, \citenamefont {Polkovnikov},\ and\
  \citenamefont {Rigol}}]{d2016quantum}%
  \BibitemOpen
  \bibfield  {author} {\bibinfo {author} {\bibfnamefont {L.}~\bibnamefont
  {D'Alessio}}, \bibinfo {author} {\bibfnamefont {Y.}~\bibnamefont {Kafri}},
  \bibinfo {author} {\bibfnamefont {A.}~\bibnamefont {Polkovnikov}},\ and\
  \bibinfo {author} {\bibfnamefont {M.}~\bibnamefont {Rigol}},\ }\bibfield
  {title} {\bibinfo {title} {From quantum chaos and eigenstate thermalization
  to statistical mechanics and thermodynamics},\ }\href@noop {} {\bibfield
  {journal} {\bibinfo  {journal} {Adv. Phys.}\ }\textbf {\bibinfo {volume}
  {65}},\ \bibinfo {pages} {239} (\bibinfo {year} {2016})}\BibitemShut
  {NoStop}%
\bibitem [{\citenamefont {Altshuler}\ \emph {et~al.}(1980)\citenamefont
  {Altshuler}, \citenamefont {Aronov},\ and\ \citenamefont
  {Lee}}]{altshuler1980interaction}%
  \BibitemOpen
  \bibfield  {author} {\bibinfo {author} {\bibfnamefont {B.~L.}\ \bibnamefont
  {Altshuler}}, \bibinfo {author} {\bibfnamefont {A.~G.}\ \bibnamefont
  {Aronov}},\ and\ \bibinfo {author} {\bibfnamefont {P.}~\bibnamefont {Lee}},\
  }\bibfield  {title} {\bibinfo {title} {Interaction effects in disordered
  fermi systems in two dimensions},\ }\href@noop {} {\bibfield  {journal}
  {\bibinfo  {journal} {Phys. Rev. Lett.}\ }\textbf {\bibinfo {volume} {44}},\
  \bibinfo {pages} {1288} (\bibinfo {year} {1980})}\BibitemShut {NoStop}%
\bibitem [{\citenamefont {Alet}\ and\ \citenamefont
  {Laflorencie}(2018)}]{alet2018many}%
  \BibitemOpen
  \bibfield  {author} {\bibinfo {author} {\bibfnamefont {F.}~\bibnamefont
  {Alet}}\ and\ \bibinfo {author} {\bibfnamefont {N.}~\bibnamefont
  {Laflorencie}},\ }\bibfield  {title} {\bibinfo {title} {Many-body
  localization: An introduction and selected topics},\ }\href@noop {}
  {\bibfield  {journal} {\bibinfo  {journal} {Comptes Rendus Physique}\
  }\textbf {\bibinfo {volume} {19}},\ \bibinfo {pages} {498} (\bibinfo {year}
  {2018})}\BibitemShut {NoStop}%
\bibitem [{\citenamefont {{\v{Z}}nidari{\v{c}}}\ \emph
  {et~al.}(2008)\citenamefont {{\v{Z}}nidari{\v{c}}}, \citenamefont {Prosen},\
  and\ \citenamefont {Prelov{\v{s}}ek}}]{vznidarivc2008many}%
  \BibitemOpen
  \bibfield  {author} {\bibinfo {author} {\bibfnamefont {M.}~\bibnamefont
  {{\v{Z}}nidari{\v{c}}}}, \bibinfo {author} {\bibfnamefont {T.}~\bibnamefont
  {Prosen}},\ and\ \bibinfo {author} {\bibfnamefont {P.}~\bibnamefont
  {Prelov{\v{s}}ek}},\ }\bibfield  {title} {\bibinfo {title} {Many-body
  localization in the {H}eisenberg {XXZ} magnet in a random field},\
  }\href@noop {} {\bibfield  {journal} {\bibinfo  {journal} {Phys. Rev. B}\
  }\textbf {\bibinfo {volume} {77}},\ \bibinfo {pages} {064426} (\bibinfo
  {year} {2008})}\BibitemShut {NoStop}%
\bibitem [{\citenamefont {Kim}\ \emph {et~al.}(2014)\citenamefont {Kim},
  \citenamefont {Chandran},\ and\ \citenamefont {Abanin}}]{kim2014local}%
  \BibitemOpen
  \bibfield  {author} {\bibinfo {author} {\bibfnamefont {I.~H.}\ \bibnamefont
  {Kim}}, \bibinfo {author} {\bibfnamefont {A.}~\bibnamefont {Chandran}},\ and\
  \bibinfo {author} {\bibfnamefont {D.~A.}\ \bibnamefont {Abanin}},\ }\bibfield
   {title} {\bibinfo {title} {Local integrals of motion and the logarithmic
  lightcone in many-body localized systems},\ }\href@noop {} {\bibfield
  {journal} {\bibinfo  {journal} {arXiv preprint arXiv:1412.3073}\ } (\bibinfo
  {year} {2014})}\BibitemShut {NoStop}%
\bibitem [{\citenamefont {Ponte}\ \emph {et~al.}(2015)\citenamefont {Ponte},
  \citenamefont {Papi{\'c}}, \citenamefont {Huveneers},\ and\ \citenamefont
  {Abanin}}]{ponte2015many}%
  \BibitemOpen
  \bibfield  {author} {\bibinfo {author} {\bibfnamefont {P.}~\bibnamefont
  {Ponte}}, \bibinfo {author} {\bibfnamefont {Z.}~\bibnamefont {Papi{\'c}}},
  \bibinfo {author} {\bibfnamefont {F.}~\bibnamefont {Huveneers}},\ and\
  \bibinfo {author} {\bibfnamefont {D.~A.}\ \bibnamefont {Abanin}},\ }\bibfield
   {title} {\bibinfo {title} {Many-body localization in periodically driven
  systems},\ }\href@noop {} {\bibfield  {journal} {\bibinfo  {journal} {Phys.
  Rev. Lett.}\ }\textbf {\bibinfo {volume} {114}},\ \bibinfo {pages} {140401}
  (\bibinfo {year} {2015})}\BibitemShut {NoStop}%
\bibitem [{\citenamefont {Zhang}\ \emph {et~al.}(2016)\citenamefont {Zhang},
  \citenamefont {Khemani},\ and\ \citenamefont {Huse}}]{zhang2016floquet}%
  \BibitemOpen
  \bibfield  {author} {\bibinfo {author} {\bibfnamefont {L.}~\bibnamefont
  {Zhang}}, \bibinfo {author} {\bibfnamefont {V.}~\bibnamefont {Khemani}},\
  and\ \bibinfo {author} {\bibfnamefont {D.~A.}\ \bibnamefont {Huse}},\
  }\bibfield  {title} {\bibinfo {title} {A floquet model for the many-body
  localization transition},\ }\href@noop {} {\bibfield  {journal} {\bibinfo
  {journal} {Phys. Rev. B}\ }\textbf {\bibinfo {volume} {94}},\ \bibinfo
  {pages} {224202} (\bibinfo {year} {2016})}\BibitemShut {NoStop}%
\bibitem [{\citenamefont {Kj{\"a}ll}\ \emph {et~al.}(2014)\citenamefont
  {Kj{\"a}ll}, \citenamefont {Bardarson},\ and\ \citenamefont
  {Pollmann}}]{kjall2014many}%
  \BibitemOpen
  \bibfield  {author} {\bibinfo {author} {\bibfnamefont {J.~A.}\ \bibnamefont
  {Kj{\"a}ll}}, \bibinfo {author} {\bibfnamefont {J.~H.}\ \bibnamefont
  {Bardarson}},\ and\ \bibinfo {author} {\bibfnamefont {F.}~\bibnamefont
  {Pollmann}},\ }\bibfield  {title} {\bibinfo {title} {Many-body localization
  in a disordered quantum ising chain},\ }\href@noop {} {\bibfield  {journal}
  {\bibinfo  {journal} {Phys. Rev. Lett.}\ }\textbf {\bibinfo {volume} {113}},\
  \bibinfo {pages} {107204} (\bibinfo {year} {2014})}\BibitemShut {NoStop}%
\bibitem [{\citenamefont {Cubitt}\ and\ \citenamefont
  {Montanaro}(2016)}]{cubitt2016complexity}%
  \BibitemOpen
  \bibfield  {author} {\bibinfo {author} {\bibfnamefont {T.}~\bibnamefont
  {Cubitt}}\ and\ \bibinfo {author} {\bibfnamefont {A.}~\bibnamefont
  {Montanaro}},\ }\bibfield  {title} {\bibinfo {title} {Complexity
  classification of local hamiltonian problems},\ }\href@noop {} {\bibfield
  {journal} {\bibinfo  {journal} {SIAM Journal on Computing}\ }\textbf
  {\bibinfo {volume} {45}},\ \bibinfo {pages} {268} (\bibinfo {year}
  {2016})}\BibitemShut {NoStop}%
\bibitem [{\citenamefont {Schuld}\ and\ \citenamefont
  {Petruccione}(2018)}]{schuld2018supervised}%
  \BibitemOpen
  \bibfield  {author} {\bibinfo {author} {\bibfnamefont {M.}~\bibnamefont
  {Schuld}}\ and\ \bibinfo {author} {\bibfnamefont {F.}~\bibnamefont
  {Petruccione}},\ }\href@noop {} {\emph {\bibinfo {title} {Supervised learning
  with quantum computers}}},\ Vol.~\bibinfo {volume} {17}\ (\bibinfo
  {publisher} {Springer},\ \bibinfo {year} {2018})\BibitemShut {NoStop}%
\end{thebibliography}%

\onecolumngrid
\appendix

\setcounter{equation}{0}%
\renewcommand{\theequation}{\thesection.\arabic{equation}}%
\setcounter{figure}{0}%
\renewcommand{\thefigure}{\thesection.\arabic{figure}}%
\setcounter{lemma}{0}%
\renewcommand{\thelemma}{\thesection.\arabic{lemma}}%
\setcounter{theorem}{0}%
\renewcommand{\thetheorem}{\thesection.\arabic{theorem}}%
\setcounter{corollary}{0}%
\renewcommand{\thecorollary}{\thesection.\arabic{corollary}}%

\section{Parameter constraint for lower bounding the gradient magnitudes by a constant}
\label{app:small_param_constant_grad}

This section presents a proof of Theorem~1 in the main text.
Under the assumption that the hardware-efficient ansatz (HEA) has a gradient component whose magnitude is constant when all parameters are zero,
Theorem~1 states that the circuit still has a gradient component with a constant magnitude when parameters satisfy a certain condition.

\subsection{Constant gradient magnitudes for the Hamiltonian variational ansatz}
In this subsection, we generalize the main results of Ref.~\cite{park2024hamiltonian}, which proved that a parameter condition such that the Hamiltonian variational ansatz (HVA) has large gradients exists.
The main difference is that we consider a circuit whose gates are generated by local operators instead of local Hamiltonians as in Ref.~\cite{park2024hamiltonian}.
Still, if we restrict operators in each layer to commute mutually, we can group them together and make a local Hamiltonian.
Using this process, we interpret the resulting circuit as the HVA and prove the same bound following Ref.~\cite{park2024hamiltonian}.
In addition, our proof here also works for a global observable, in contrast to Ref.~\cite{park2024hamiltonian}, which only considered a local observable.

Let us consider a parameterized quantum circuit for a system with $N$ qubits, given by
\begin{align}
    \ket{\psi(\pmb{\theta})} = \lprod_{i=1}^D \bigl[ \prod_{j=1}^{M_i} e^{-i G^{(i,j)} \theta_{i,j}} \bigr] \ket{\psi_0}, \label{eq:generalized_paramerized_circuit}
\end{align}
where generators $G^{(i,j)}$ for each $i$ are mutually commuting, i.e., $[G^{(i,j)},G^{(i,j')}]=0$ for all $i,j,j'$.
We also assume that each $G^{(i,j)}$ is a local operator, acting on at most $\Theta(1)$ sites and $M_i = \mathcal{O}(N)$.
Here, two notations $\prod$ and $\lprod$ are defined as 
\begin{gather}
    \prod_{i=1}^k U_i := U_1 \cdots U_k, \qquad \lprod_{i=1}^k U_i = U_k\cdots U_1.
\end{gather}
The circuit given in Eq.~\ref{eq:generalized_paramerized_circuit} can be considered a generalized version of the HVA.

We now consider the cost function is given by
\begin{align}
    C = \braket{\psi(\pmb{\theta}) | O | \psi(\pmb{\theta})}.
\end{align}
A component of the gradient for $C$ is obtained as
\begin{align*}
    \partial_{n,m}C &:= \frac{\partial C}{\partial \theta_{n,m}} \\
    & = \Bigl\langle \psi_0 \Bigl \vert \prod_{i=1}^n \bigl[ \prod_{j=1}^{M_i} e^{i G^{(i,j)} \theta_{i,j}} \bigr] (iG^{(n,m)}) \prod_{i=n}^D \bigl[ \prod_{j=1}^{M_i} e^{i G^{(i,j)} \theta_{i,j}} \bigr] O \lprod_{i=1}^D \bigl[ \prod_{j=1}^{M_i} e^{-i G^{(i,j)} \theta_{i,j}} \bigr] \vert \psi_0 \Bigr \rangle \\
    &\quad + \Bigl\langle \psi_0 \Bigl \vert \prod_{i=1}^D \bigl[ \prod_{j=1}^{M_i} e^{i G^{(i,j)} \theta_{i,j}} \bigr] O \lprod_{i=n}^D \bigl[ \prod_{j=1}^{M_i} e^{-i G^{(i,j)} \theta_{i,j}} \bigr] (-iG^{(n,m)}) \lprod_{i=1}^n \bigl[ \prod_{j=1}^{M_i} e^{-i G^{(i,j)} \theta_{i,j}} \bigr] \Bigr \vert \psi_0 \Bigr \rangle \\
    &= i \braket{\psi_0 | U_B^\dagger  [G_{n,m}, U_A^\dagger O U_A] U_B | \psi_0},
\end{align*}
where
\begin{gather}
    U_A = \lprod_{i=n}^D \bigl[ \prod_{j=1}^{M_i} e^{-i G^{(i,j)} \theta_{i,j}} \bigr], \qquad U_B = \lprod_{i=1}^n \bigl[ \prod_{j=1}^{M_i} e^{-i G^{(i,j)} \theta_{i,j}} \bigr] \label{eq:def_u_a_b}.
\end{gather}

Under this setup, we have the following lemma.
\begin{lemma}\label{lm:grad_speed_limit}
    For $\rho_0 = \ket{\psi_0}\bra{\psi_0}$, let us assume that $|\Tr[\rho_0 [G^{(n,m)},O]]| > 0$, and there exist Hamiltonians $H_A$, $H_B$ such that $U_A = e^{-i H_A t_A}$ and $U_B = e^{-i H_B t_B}$ for some $t_A,t_B \geq 0$. Then,
    \begin{align}
        |\partial_{n,m}C| \geq |\Tr[\rho_0 [G^{(n,m)},O]]|/2
    \end{align}
    for $t_A+t_B \leq t_c  := |\Tr[\rho_0 [G^{(n,m)},O]]|/(4KQ)$, where $K=\max\{\Vert H_B \Vert, \Vert [H_A, O] \Vert\}$ and $Q=\max\{\Vert [G^{(n,m)},O] \Vert, \Vert G^{(n,m)} \Vert \}$.
    Here, $\Vert \cdot \Vert$ is the operator norm.
\end{lemma}
\begin{proof}
Let 
\begin{align}
    A(t_1, t_2) = i\Tr [ e^{-iH_B t_2} \rho_0 e^{i H_B t_2} [G^{(n,m)}, e^{iH_A t_1} O e^{-iH_A t_1}]].
\end{align}
One can see that $A(0, 0) = i\Tr[\rho_0 [G^{(n,m)},O]]$ and $A(t_A, t_B) = \partial_{n,m} C$.

Then, 
\begin{align}
    &|A(t_A, t_B) - A(0,0)| \leq \int_{0}^{t_A} dt_1 \Bigl| \frac{\partial A(t_1, t_B)}{\partial t_1}  \Bigr| + \int_{0}^{t_B} dt_2 \Bigl| \frac{\partial A(0, t_2)}{\partial t_2}  \Bigr|.
\end{align}
We further have
\begin{align*}
    \Bigl| \frac{d A(0,t_2)}{\partial t_2} \Bigr| &= \Bigl|\Tr\bigl\{ [H_B, \rho_0(t_2)] [G^{(n,m)}, O] \bigr\} \Bigr| \\
    &\leq 2 \Vert H_B \Vert \Vert [G^{(n,m)}, O] \Vert \leq 2KQ, \numberthis
\end{align*}
and 
\begin{align*}
    \Bigl| \frac{d A(t_1,t_B)}{\partial t_1} \Bigr| &= \Bigl|\Tr\bigl\{ \rho_0(t_B) [G^{(n,m)}, [H_A, e^{i H_A t_1} O e^{-i H_A t_1}]] \bigr\} \Bigr|\\
    &\leq 2 \Vert G^{(n,m)} \Vert \Vert [H_A, O] \Vert \leq 2KQ, \numberthis
\end{align*}
where $\rho_0(t) = e^{-iH_B t} \rho_0 e^{i H_B t}$.

Integrating both sides, we have
\begin{align}
    |A(t_R, t_L) - A(0, 0)| \leq 2 KQ(t_R + t_L). \label{eq:partial_deriv_bound}
\end{align}
By entering $t_R+t_L \leq t_c = |A(0, 0)|/(4KQ)$, we obtain the desired inequality.
\end{proof}
Note that we used different definitions of $K$ and $Q$ from Ref.~\cite{park2024hamiltonian} to incorporate that $G^{(n,m)}$ is a local operator (instead of a sum of local operators considered in Ref.~\cite{park2024hamiltonian}).

Although exact values of $K$ and $Q$ depend on the definitions of $H_{A}$, $H_B$, and $O$, the scaling of $K$ and $Q$ can be obtained under the reasonable assumptions that $H_A, H_B$ are local Hamiltonians (sums of operators acting on at most $\mathcal{O}(1)$ nearby sites in a given lattice) and $G^{(n,m)}$ is a local operator.
In this case, one can readily see that $K = \Theta(N)$, $Q=\Theta(1)$ both for $O$ given by (1) a Pauli string and (2) a local Hamiltonian.

Next, we find those Hamiltonians $H_A$ and $H_B$ for the parameterized circuit given by Eq.~\eqref{eq:generalized_paramerized_circuit}.
For this purpose, we rewrite each layer of the circuit. We introduce a Hamiltonians $H^{(i)}$ defined as
\begin{align}
    H^{(i)} := \sum_{j=1}^{M_i} G^{(i,j)} \frac{\theta_{i,j}}{\theta^{(i)}_{\rm max}}, \label{eq:ham_layer_i}
\end{align}
where $\theta^{(i)}_{\rm max} = \max_{j} \theta_{i,j}$. Using this Hamiltonian, our circuit can be rewritten as
\begin{align}
    U = \lprod_{i=1}^D e^{-i H^{(i)} \theta^{(i)}_{\rm max}}. \label{eq:generalized_paramerized_circuit_ham}
\end{align}

We also define two quantities, $H_{\rm max}$ and $J$, for the truncated Floquet-Magnus~\cite{kuwahara2016floquet} expansion.
First, $H_{\rm max}$ upper bounds the norm of the Hamiltonian:
\begin{align}
    H_{\rm max} = \max_{i} \Vert H^{(i)} \Vert
\end{align}
Second, $J$ upper bounds the local interaction strength. Let $\mathrm{supp}(O)$ be the set of sites that $O$ acts on. For example, $\mathrm{supp}(X_1X_2X_N) = \{1,2,N\}$.
Then, $J$ is defined by
\begin{align}
    J=\max_{i} \max_{a \in \Lambda} \sum_{j: \mathrm{supp}(G^{(i,j)}) \ni a} \left\Vert G^{(i,j)} \frac{\theta_{i,j}}{\theta^{(i)}_{\rm max}} \right\Vert,
\end{align}
where $\Lambda = [N] := \{1, \cdots , N\}$ is the set of all sites.

Finally, we assume that $G^{(i,j)}$ acts at most $k$ sites, i.e., $\max_{i,j} |\mathrm{supp}(G^{(i,j)})| \leq k$.
Then, the truncated Floquet-Magnus expansion~\cite{kuwahara2016floquet} gives the following result:
\begin{lemma}[Proposition~3 in Ref.~\cite{park2024hamiltonian}]\label{lm:truncated_fm_expansion}
    Let $U_A, U_B$ be the unitary operators defined in Eq.~\eqref{eq:def_u_a_b}.
    For the parameters $H_{\rm max}$ and $J$ defined above, we can find a $(r+1)k$-local Hamiltonian $H_A^{(r)}$ such that 
    \begin{align}
        \Bigl\Vert U_A - e^{-iH_A^{(r)} t_A}\Bigr\Vert \leq 6 H_{\rm max} 2^{-r_0} t_A + \frac{2 H_{\rm max}(2 k J)^{r+1}}{(r+1)^2} (r+1)! t_A^{r+2}, \label{eq:truncated_fm_error}
    \end{align}
    with $t_A = \sum_{i=n}^D \theta^{(i)}_{\rm max}$ for all $r \leq r_0 := \lfloor 1/(32 k J t_A) \rfloor$.
    Also, the same inequality holds for $U_B$ with $H_B$ and $t_B = \sum_{i=1}^n\theta^{(i)}_{\rm max}$.
\end{lemma}
This is a direct application of the main result of Ref.~\cite{kuwahara2016floquet} to Eq.~\eqref{eq:generalized_paramerized_circuit_ham}.

One can prove the following theorem by combining Lemmas~\ref{lm:grad_speed_limit} and \ref{lm:truncated_fm_expansion}.
\begin{theorem}\label{thm:smal_param_constant_grad_bound_app}
    For a circuit defined in Eq.~\eqref{eq:generalized_paramerized_circuit}, assume that $O$ is a Pauli-string or a $k$-local Hamiltonian, $g:=|\Tr[\rho_0 [G^{(n,m)}, O]]| > 0$, and $J=\mathcal{O}(1)$.
    Then, there exists $\gamma > 0$ such that 
    \begin{align}
        |\partial_{n,m}C| \geq g/4
    \end{align}
    is satisfied for sufficiently large $N$ if $\sum_{i=1}^D \theta^{(i)}_{\rm max} \leq \gamma / N$.
\end{theorem}

\begin{proof}
First, from Lemma~\ref{lm:truncated_fm_expansion}, we can set $r=1$ if $t_A \leq (32 k J)^{-1}$.
Then, the error from the truncated Floquet-Magus expansion, given by the RHS of Eq.~\eqref{eq:truncated_fm_error}, can be lower bounded by $\mathcal{O}(H_{\rm max} t_A^3)$.
In addition, from the definition of $H^{(i)}$ (see Eq.~\eqref{eq:ham_layer_i}), we obtain $K=\mathcal{O}(N)$ for Lemma~\ref{lm:grad_speed_limit}.
Likewise, since $G^{(i,j)}$ are local operators, we have $Q=\mathcal{O}(1)$.
With the lemma given below, we can find a constant $\gamma_1 > 0$ such that the following error bound is satisfied for $t_{A,B} \geq 0$ and $t_A+t_B \leq \gamma_1/N$:
\begin{align}
    \partial_{n,m} C = g/2 + \mathcal{O}(\epsilon \Vert G^{(n,m)} \Vert \Vert O \Vert ),
\end{align}
where $g:=|\Tr[\rho_0 [G^{(n,m)}, O]]|$ and $\epsilon$ is the RHS of Eq.~\eqref{eq:truncated_fm_error}.
Given that we can also find $\gamma_2 >0$ such that $\epsilon = \mathcal{O}(1/N^2)$ for $t_A + t_B \leq \gamma_2/N$, the desired result can be proven for sufficiently large $N$ with $\gamma = \min\{\gamma_1, \gamma_2\}$.
\end{proof}
See also Ref.~\cite{park2024hamiltonian} for more detailed error analysis (without big-$\mathcal{O}$ notations).

\begin{lemma}
    Suppose that $\rho$ is a density matrix satisfying $\rho \geq 0$ and $\Tr[\rho]=1$ and $A$ is an  Hermitian operator, and $U_1$, $U_2$, $\tilde{U}_1$, $\tilde{U}_2$ are unitary operators. We further assume that $|U_1 - \tilde{U}_1| \leq \epsilon$ and $|U_2 - \tilde{U}_2| \leq \epsilon$. Then, the following inequality holds:
    \begin{align}
    \bigl| i\Tr[U_1 \rho U_1^\dagger [A, U_2^\dagger O U_2]] - i\Tr[\tilde{U}_1 \rho \tilde{U}_1^\dagger [A, \tilde{U}_2^\dagger O \tilde{U}_2]] \bigr| \leq 8 \epsilon \Vert A \Vert \Vert O\Vert. \label{eq:grad_diff_bound}
    \end{align}
\end{lemma}
\begin{proof}
We first obtain the following inequality:
\begin{align*}
    \bigl| \Tr[U \rho U^\dagger B] - \Tr [\tilde{U} \rho \tilde{U}^\dagger \tilde{B}] \bigr| &= \bigl| \Tr[\rho (U^\dagger B U - \tilde{U}^\dagger \tilde{B}\tilde{U})] \bigr|\\
    &\leq \Vert U^\dagger B U - \tilde{U}^\dagger \tilde{B}\tilde{U} \Vert = \Vert U^\dagger B U - U^\dagger B \tilde{U} + U^\dagger B \tilde{U} - \tilde{U}^\dagger \tilde{B}\tilde{U} \Vert \\
    &\leq \Vert B \Vert \Vert U - \tilde{U} \Vert + \Vert U^\dagger B - \tilde{U}^\dagger \tilde{B} \Vert \\
    &\leq \Vert B \Vert \Vert U - \tilde{U} \Vert + \Vert B - \tilde{B} \Vert + \Vert \tilde{B} \Vert \Vert U^\dagger - \tilde{U}^\dagger \Vert.
\end{align*}
Entering $U=U_1$, $\tilde{U}=\tilde{U}_1$, $B=i[A, U_2^\dagger O U_2]$, and $B=i[A, \tilde{U}_2^\dagger O \tilde{U}_2]$ to the above inequality yields
\begin{align}
    \bigl| i\Tr[U_1 \rho U_1^\dagger [A, U_2^\dagger O U_2]] - i\Tr[\tilde{U}_1 \rho \tilde{U}_1^\dagger [A, \tilde{U}_2^\dagger O \tilde{U}_2]] \bigr| \leq 4 \epsilon \Vert A \Vert \Vert O \Vert + \Vert B - \tilde{B} \Vert.
\end{align}
In addition, 
\begin{align}
    \Vert B - \tilde{B} \Vert = \Vert [A, (U_2^\dagger O U_2 - \tilde{U}_2^\dagger O \tilde{U}_2)] \Vert \leq 4 \Vert A \Vert \Vert O \Vert \Vert U_2 - \tilde{U}_2 \Vert.
\end{align}
Combining all these inequalities, we obtain the desired result.
\end{proof}

We conclude this subsection with a remark on the scaling of $J$.
From the definition of $H^{(i)}$ given in Eq.~\eqref{eq:ham_layer_i}, $J=\Theta(1)$ is naturally obtained without additional assumption when each pair of $G^{(i,j)}$ and $G^{(i,j')}$ overlaps a finite number of sites, i.e., $|\mathrm{supp}(G^{(i,j)}) \cap \mathrm{supp}(G^{(i,j')})| = O(1)$ regardless of $i,j,j'$.
This condition is naturally satisfied for circuits with geometrically local connectivity (see the next subsection).

\subsection{Converting the hardware efficient ansatz to a circuit with parameterized entangling gates}
The main limitation of the results in the previous subsection is that a circuit must be given by Eq.~\eqref{eq:generalized_paramerized_circuit}, all gates of which are parameterized.
Thus, to apply Theorem~\ref{thm:smal_param_constant_grad_bound_app} to the hardware efficient ansatz (HEA), defined in the main text, we need to remove non-parameterized entangling gates from the HEA.
Recall the definition of HEA given in the main text, given by $U(\pmb{\theta})=\lprod_{i=1}^p V(\pmb{\theta}_{i,:})$ where
\begin{align}
    V(\pmb{\theta}_{i,:}) = \prod_{\braket{j, j'} \in E} \mathrm{CZ}_{j,j'}\prod_{j=1}^N e^{-i Z_j \theta_{i,j+N}/2} \prod_{j=1}^N e^{-i X_j \theta_{i,j}/2}.
\end{align}

In the main text, we introduced a circuit identity given by 
\begin{align}
    \prod_{\braket{j, j'} \in E} \mathrm{CZ}_{j,j'} \prod_{j=1}^N e^{-i X_j \theta_{i,j}/2} = \prod_{j=1}^N e^{-i \Lambda_j \theta_{i,j}/2} \prod_{\braket{j, j'} \in E} \mathrm{CZ}_{j,j'}
\end{align}
where $\Lambda_j = X_j \prod_{l \in \mathcal{N}(j)}$ and $\mathcal{N}(j) = \{j': \braket{j,j'} \in E\}$ is the neighbors of $j$ in a given interaction graph.
Our identity follows from
\begin{align}
    \prod_{\braket{j, j'} \in E} \mathrm{CZ}_{j,j'} X_k \prod_{\braket{j, j'} \in E} \mathrm{CZ}_{j,j'} = X_k \prod_{l \in \mathcal{N}(k)} Z_l.
\end{align}

We now use the above identity to remove the CZ gates from the circuit.
Let us first consider the case where $p$ is even.
In this case, we move the CZ gates for each $2i$-th block to the front.
Precisely, we have
\begin{align*}
    &V(\pmb{\theta}_{2i,:})V(\pmb{\theta}_{2i-1,:}) \\
    &= \prod_{\braket{j, j'} \in E} \mathrm{CZ}_{j,j'}\prod_{j=1}^N e^{-i Z_j \theta_{2i,j+N}/2} \prod_{j=1}^N e^{-i X_j \theta_{2i,j}/2}\prod_{\braket{j, j'} \in E} \mathrm{CZ}_{j,j'}\prod_{j=1}^N e^{-i Z_j \theta_{2i-1,j+N}/2} \prod_{j=1}^N e^{-i X_j \theta_{2i-1,j}/2}\\
    &= \prod_{j=1}^N e^{-i Z_j \theta_{2i,j+N}/2} \prod_{\braket{j, j'} \in E} \mathrm{CZ}_{j,j'} \prod_{j=1}^N e^{-i X_j \theta_{2i,j}/2}\prod_{\braket{j, j'} \in E} \mathrm{CZ}_{j,j'}\prod_{j=1}^N e^{-i Z_j \theta_{2i-1,j+N}/2} \prod_{j=1}^N e^{-i X_j \theta_{2i-1,j}/2}\\
    &= \prod_{j=1}^N e^{-i Z_j \theta_{2i,j+N}/2}  \prod_{j=1}^N e^{-i \Lambda_j \theta_{2i,j}/2} \prod_{j=1}^N e^{-i Z_j \theta_{2i-1,j+N}/2} \prod_{j=1}^N e^{-i X_j \theta_{2i-1,j}/2}.
\end{align*}
Thus, the HEA with $p$ blocks can be converted to the circuit given by Eq.~\eqref{eq:generalized_paramerized_circuit} with $D=2p$.
In addition, each generator acts on at most $k=1+\max_{i \in [n]}|\mathcal{N}(i)|$ sites where $[n]=\{1,\cdots,N\}$ is a set of all qubits.
When the HEA is defined on a finite-dimensional lattice, $k$ is also constant (independent of $N$).
Thus, Theorem~\ref{thm:smal_param_constant_grad_bound_app} can be directly applied to the resulting circuit.

On the other hand, when $p$ is odd, we move the CZ gates for each $2i-1$-th block to the front.
In this case, the resulting circuit will be
\begin{align*}
    U(\pmb{\theta}) &= \lprod_{i=1}^{\lfloor \frac{p}{2}\rfloor} \Bigl[ \prod_{j=1}^N e^{-i Z_j \theta_{2i+1,j+N}/2}  \prod_{j=1}^N e^{-i \Lambda_j \theta_{2i+1,j}/2} \prod_{j=1}^N e^{-i Z_j \theta_{2i,j+N}/2} \prod_{j=1}^N e^{-i X_j \theta_{2i,j}/2} \Bigr] \\
    &\quad \times \prod_{j=1}^N e^{-i Z_j \theta_{1,j+N}/2}  \prod_{j=1}^N e^{-i \Lambda_j \theta_{1,j}/2} \prod_{\braket{j, j'} \in E} \mathrm{CZ}_{j,j'}.
\end{align*}

We now define a modified initial state $\rho'=\prod_{\braket{j, j'} \in E} \mathrm{CZ}_{j,j'} \ket{\psi_0}\bra{\psi_0} \prod_{\braket{j, j'} \in E} \mathrm{CZ}_{j,j'}$ and the circuit $U'(\pmb{\theta})$ without the first CZ layer.
Note that the number of layers in $U'(\pmb{\theta})$ is $D=2p$.
Then, $C = \braket{\psi_0 | U(\pmb{\theta})^\dagger O U(\pmb{\theta}) | \psi_0} = \Tr[ O U'(\pmb{\theta}) \rho' U'(\pmb{\theta})^\dagger]$ and we can apply Theorem~\ref{thm:smal_param_constant_grad_bound_app} to $U'(\pmb{\theta})$.
Moreover, as $U(\pmb{\theta}=0)=\prod_{\braket{j, j'} \in E} \mathrm{CZ}_{j,j'}$, we obtain
\begin{align}
    g := |\Tr[\rho' [G^{(n,m)},O]]| = |\Tr[U(0)^\dagger \rho U(0) [G^{(n,m)},O]]| = |\partial_{n,m}C |_{\pmb{\theta}=0}.
\end{align}

In summary, we proved the following theorem.
\begin{theorem}[Restatement of Theorem~1 in the main text] \label{thm:hea_small_param_large_grad_app}
    Let $C(\pmb{\theta}) = \braket{\psi(\pmb{\theta})|O|\psi(\pmb{\theta})}$ where $\ket{\psi(\pmb{\theta})}=U(\pmb{\theta})\ket{\psi_0}$ be the cost function.
    Assume that $O$ is either a Pauli-string or $k$-local Hamiltonian, and there exist $n,m$ such that $|\partial_{n,m} C |_{\pmb{\theta}=0} = \Omega(1)$.
    Then, there exists a constant $\gamma > 0$ such that $|\partial_{n,m} C| = \Omega(1)$ when $0 \leq \theta_{i,j} \leq \gamma / (pN)$ is satisfied for all $i$ and $j$.
\end{theorem}

\begin{remark}
    Our technique that converts the HEA to Eq.~\eqref{eq:generalized_paramerized_circuit} works for arbitrary Clifford entangling gates.
\end{remark}

\section{Floquet many-body localization in the hardware-efficient ansatz}
\label{app:floquet-mbl}

In this section, we study conditions when the HEA is in the many-body localized phase.

\subsection{Brief introduction to many-body localization}
Chaotic quantum many-body systems thermalize in the sense that the time average of an observable is the same as its thermal ensemble average (see Refs.~\cite{eisert2015quantum,gogolin2016equilibration} and Refs.~\cite{nandkishore2015many,d2016quantum} for reviews from quantum information and statistical mechanics viewpoints, respectively).
Chaotic systems are often characterized by the level statistics following the Gaussian orthogonal (when the system is time-reversal symmetric) or Gaussian unitary (otherwise) ensemble.
While typical quantum many-body systems are chaotic, there are some counterexamples.
Integrable and Anderson localized systems are the two most well-known traditional non-chaotic systems.
In an integrable system, conserved quantities can be computed analytically, and the energy levels show the Poisson statistics.
As the system has an extensive number of independent conserved quantities, integrable systems do not thermalize to the Gibbs ensemble.
Anderson localized systems are another example with an extensive number of conserved quantities.
Thermalization in these systems is prevented by disorders, and local excitation does not spread over a system.

Recently, many-body localized (MBL) systems have been widely studied as examples of a stable non-chaotic phase.
These systems were first introduced in 1980 by Altshuler et al.~\cite{altshuler1980interaction} as a perturbation to an Anderson localized system, but have gained lots of interest in recent decades as the advances in numerical techniques could reveal interesting properties of the system (see, e.g., Ref.~\cite{alet2018many} for a review).
Similar to the Anderson localization, local disorders are the main ingredients that prevent MBL systems from thermalization.
However, multiple excitations in an MBL system interfere with each other and induce dephasing.
Such interference leads to the logarithmic growth of entanglement~\cite{vznidarivc2008many}, which is a unique property of MBL systems.
In contrast, the entanglement of an Anderson localized system does not grow at all, and that of an integrable system grows linearly.


Information theoretically, MBL systems are characterized by a logarithmic lightcone~\cite{kim2014local}.
For local operators $O_A$ and $O_B$ acting on subsystems $A$ and $B$, respectively, the MBL system satisfies
\begin{align}
    \mathbb{E}_\mu \Vert [O_A, O_B(t)] \Vert \leq c t e^{-\mathrm{dist}(A,B)/\xi},
\end{align}
where the average is taken over the distribution of the disorders, $c$ is a constant, $\mathrm{dist}(A,B)$ is the distance between $A$ and $B$ for a given lattice, and $\xi$ is the localization length.
This contrasts general local many-body Hamiltonians having a linear lightcone, satisfying
\begin{align}
    \Vert [O_A, O_B(t)] \Vert \leq c \min(|A|, |B|) e^{-a(\mathrm{dist}(A,B)-vt)},
\end{align}
where $|A|$ and $|B|$ are the size of the subsystems, $a$ is a constant, and $v$ is the Lieb-Robinson velocity. 
The Lieb-Robinson velocity quantifies the speed of information propagation in a given system.
Many important properties of MBL systems, such as the absence of transport and the slow growth of entanglement, can be explained using the logarithmic lightcone~\cite{kim2014local}.

A phenomenological model~\cite{huse2014phenomenology} provides an alternative view to understand MBL systems.
In this model, an MBL system is described using quasi-local conserved quantities $\{\tau^{i}_z\}_{i=1}^N$.
Precisely, we expect that the Hamiltonian is written in terms of these operators as
\begin{align}
    H = \sum_{i} J_i \tau_{i}^z + \sum_{i \neq j} J_{ij} \tau_{i}^z \tau_{j}^z +\sum_{\text{all distinct }i,j,k } J_{ijk} \tau_{i}^z \tau_{j}^z \tau_{k}^z +\cdots, \label{eq:app_mbl_phenom}
\end{align}
where the strengths of the many-body interactions $\{J_{ij},J_{ijk},\cdots\}$ decay exponentially with the distance between the sites that the interaction acts on.
Formally, we can write that $J_S \propto e^{-d/\xi}$ where $S \subset [n]$ is a subset of $[n] = \{1,\cdots,N\}$, $d = \max_{i,j \in S}\mathrm{dist}(i,j)$ is the maximum distance between sites in $S$, and $\xi$ is the localization length.
In addition, each $\tau^{i}_z$ has an overlap with $Z_i$ by a constant, i.e., $\tau^{i}_z = a Z_i + \cdots$ where $a$ is a constant independent to $N$.
In other words, there exists a unitary operator $W$ that transforms $\tau^i_z$ to $Z_i$, i.e.,
\begin{align}
    W^\dagger \tau^i_z W = Z_i,
\end{align}
and $W$ is described by a local short-depth circuit.
As a consequence, $W^\dagger H W$ is diagonal in the $Z$-basis, and $W \ket{x}$ becomes the eigenstate of $H$ for any product state $\ket{x}$ in the computational basis.

The MBL phase can also be found in a periodically-driven system~\cite{ponte2015many,zhang2016floquet}.
In this case, all eigenstates of the Floquet operator $U(T)=\mathcal{T}[e^{-i \int_0^T dt H(t)}]$ follow the area law, and $U(nT):=U(T)^n$ shows a logarithmic lightcone.
We also expect that there is an effective Hamiltonian $H_{\rm eff}$ such that $U(T) = e^{-i H_{\rm eff} T}$ and can be written in the form of Eq.~\eqref{eq:app_mbl_phenom}.

\begin{figure*}
    \centering
    \includegraphics[width=0.95\linewidth]{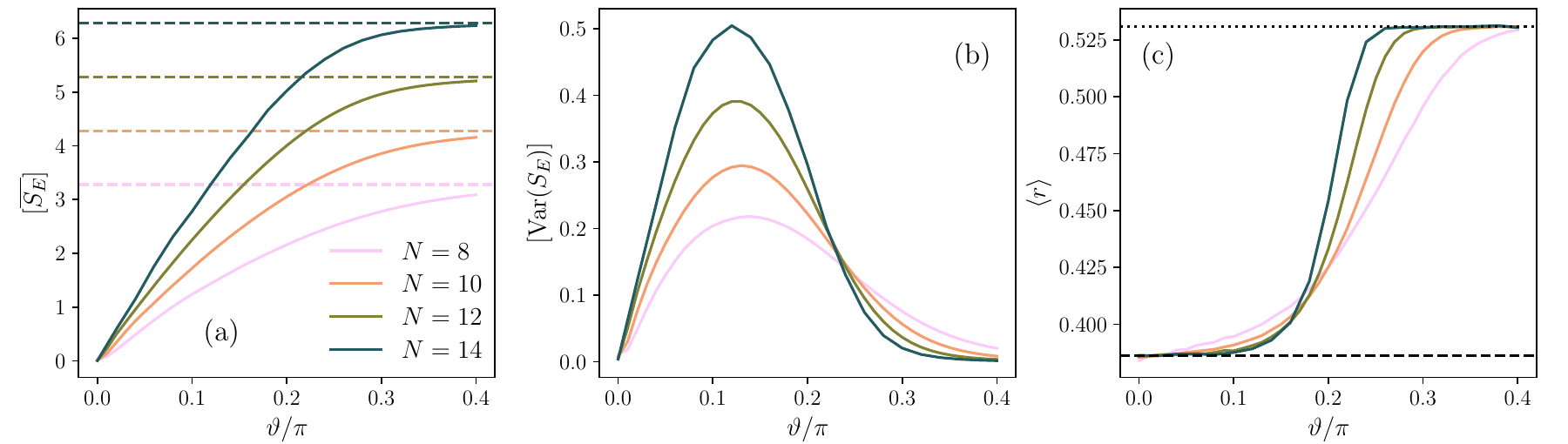}
    \caption{Many-body localization of a unitary operator $\tilde{V}(\theta)$. (a) Half-chain entanglement entropy for eigenstates of $\tilde{V}(\theta)$ as a function of $\theta/\pi$. Results are averaged over all eigenstates and disorder realizations. Dashed horizontal lines indicate the Page entropy, which is expected for Haar random states.
    (b) Variance of the eigenstate entanglement entropy averaged over disorder realizations. For each random instance of $\tilde{V}(\theta)$, we compute $\overline{S_E^2} - \overline{S_E}^2$, and the results are averaged over all instances.
    (c) The averaged adjacent gap ratios. For ordered quasi-energy levels $\{E_i\}$ for each random instance of $\tilde{V}(\theta)$, gaps $\Delta_i = E_{i+1}-E_i$ are obtained. Then, the ratios $r_i=\min\{\Delta_{i+1}/\Delta_i,\Delta_{i}/\Delta_{i+1}\}$ are averaged over $i$ and all random instances.
    Horizontal lines indicate the expected averaged values of $r$ for the Possion (dashed) and the Gaussian orthogonal ensemble (dotted).
    All presented results are obtained from $2^{12}$ random instances for $N \in [8, 10]$, $2^{10}$ for $N=12$, and $2^7$ for $N=14$.}
    \label{fig:mbl_transition}
\end{figure*}

\subsection{Many-body localized hardware-efficient ansatz}
We interpret each block of the 1D HEA as a Floquet operator and study the phases of this operator.
Precisely, we investigate the phases of a unitary operator given by
\begin{align}
    V(\vartheta) = \prod_{j=0}^{N-1} \mathrm{CZ}_{j,j+1} \prod_{j=1}^N e^{-i Z_j \phi_{j}/2} \prod_{j=1}^N e^{-i X_j \vartheta/2},
\end{align}
where each $\phi_j$ is randomly sampled from the uniform distribution between $-\pi$ and $\pi$.
This is the same as each block of the HEA, $V(\pmb{\theta}_{i,:})$, considered in the main text besides that we assign $\theta_{i,j} = \vartheta$ for all $1 \leq j \leq N$.

Note that a recent study by Shtanko et al.~\cite{shtanko2023uncovering} already has shown that the HEA can have the MBL phase within a certain parameter condition for the one-dimensional (theoretically) and the heavy-hexagonal (experimentally) lattices.
However, as the circuit considered in Ref.~\cite{shtanko2023uncovering} is slightly different from ours, we investigate the MBL phase of our circuit model in this subsection.

When $\vartheta=0$, $V$ is diagonal in the $Z$-basis, and all eigenstates of $V$ are product states. This is a characteristic of a fully localized system.
If $\vartheta$ is non-zero but small, all eigenstates are still very close to product states, which is a signature of the MBL phase.
As we increase $\vartheta$, the off-diagonal terms of $V$ also increase, and at a certain value of $\vartheta=\vartheta_c$, the system will become chaotic.
We study such a transition numerically.

For numerical study, we use a shifted version of $V(\vartheta)$, which is given by
\begin{align}
    \tilde{V}(\vartheta) = \prod_{j=1}^N e^{-i X_j \vartheta/4} \prod_{j=0}^{N-1} \mathrm{CZ}_{j,j+1} \prod_{j=1}^N e^{-i Z_j \phi_{j}/2} \prod_{j=1}^N e^{-i X_j \vartheta/4}.
\end{align}
As all eigenstates of $\tilde{V}(\vartheta)$ are real-valued since $\tilde{V}(\vartheta)^T = \tilde{V}(\vartheta)$, numerical diagonalization of $\tilde{V}(\vartheta)$ is more feasible than the original unitary operator $V(\vartheta)$.

To obtain the phase diagram, we utilize the diagnostics developed in Refs.~\cite{kjall2014many,ponte2015many,zhang2016floquet}.
For each random instance of $\tilde{V}(\vartheta)$, we compute the eigenstates and corresponding quasi-energies, defined by
\begin{align}
    \tilde{V}(\vartheta) = \sum_{i=1}^{2^N} e^{-iE_i} \ket{E_i} \bra{E_i},
\end{align}
where each $-\pi \leq E_i \leq \pi$ is a quasi-energy and $\ket{E_i}$ is an eigenstate.
For each eigenstate of $\ket{E_i}$, we compute the half-chain entanglement entropy.
We divide the chain into two subsystems $A=[1,\cdots,N/2]$ and $B=[N/2,\cdots,N]$.
Then, for $\rho_A = \Tr_B[\ket{E_i} \bra{E_i}]$, the entanglement entropy $S_{E_i} = -\Tr[\rho_A \log_2 \rho_A]$ is computed numerically.

We use the following notations for our numerical results.
An overbar indicates the average over eigenstates of each instance of $\tilde{V}(\theta)$, and a bracket is used for the average over disorder realizations.
For example, $\overline{S_E} = 2^{-N} \sum_{i=1}^{2^N} S_A(\ket{E_i})$ where $\ket{E_i}$ is an eigenstate of $V(\theta)$ with quasi-energy $E_i$.

Our first diagnostic is the entanglement of entropy itself.
When a system is chaotic, the entanglement entropy of each eigenstate is close to that of Haar random states, given by the Page entropy $S_{\rm Page} = N/2 - \log_2(e)/2$.
We plot the entanglement entropy averaged over all eigenstates and the disorder realizations, $[\overline{S_E}]$, in Fig.~\ref{fig:mbl_transition}(a).
We observe that the entanglement entropy gets closer to the Page entropy as $\vartheta$ increases, which shows the existence of the chaotic phase.
However, the entanglement entropy does not tell much about the transition point.

To study the transition point, we plot the variance of entanglement entropy, $\mathrm{Var}(S_E) = \overline{S_E^2}-\overline{S_E}^2$, averaged over the disorder realizations in Fig.~\ref{fig:mbl_transition}(b).
The variance indicates the transition point~\cite{kjall2014many}. This is because, near the MBL transition, some eigenstates follow the area law, but others follow the volume law.
For $N=14$, the plot shows the maximum variance is obtained when $\vartheta/\pi \approx 0.13$.

We next study the adjacent gap ratios.
For ordered quasi-energies $\{E_i\}$, we compute its gap $\Delta_i = E_{i+1}-E_i$ and their ratios $r_i=\min\{\Delta_{i+1}/\Delta_i,\Delta_{i}/\Delta_{i+1}\}$.
We then compute $\braket{r}$, the average of $r_i$ for all $i$ and the disorder realizations.
When the system is in a localized phase, we expect $\braket{r} \approx 0.39$, which is the value expected for the Poisson distribution.
On the other hand, $\braket{r} \approx 0.53$ in a chaotic phase, which is from the Gaussian orthogonal ensemble (GOE).
The averaged ratios, $\braket{r}$, for $N\in[8, 10, 12, 14]$ are plotted as functions of $\vartheta/\pi$ in Fig.~\ref{fig:mbl_transition}(c).
We observe that $\braket{r}$ is close to that of the Poisson distribution when $\vartheta$ is small and becomes that of GOE when $\vartheta/\pi \gtrsim 0.25$.
In addition, the plots of $\braket{r}$ for $N\in[10,12,14]$ cross near $\vartheta/\pi \approx 0.16$.

In summary, the system is in the MBL and the chaotic phases when $\vartheta < \vartheta_c$ and $\vartheta > \vartheta_c$, respectively.
The phase transition point is $0.13 \lesssim \vartheta_c / \pi \lesssim 0.16$.

\subsection{Product of Floquet-MBL systems is an MBL system}
In the previous subsection, we studied the phase of a single Floquet operator $V(\vartheta)$.
However, the circuit we used in the main text is a product of $V(\vartheta)$ with different values of $\vartheta$, given by $U=V(\vartheta_p)\cdots V(\vartheta_1)$.
It is less obvious whether such $U$ must be in the MBL phase (recall that a product unitary operators generated by time-independent Hamiltonians does not necessarily conserve the energy).
In this subsection, we argue that $U$ is also in the MBL phase under the following conjecture:

\begin{conjecture}\label{conj:log_lightncone_mbl}
    A logarithmic lightcone is a sufficient condition for the MBL phase.
\end{conjecture}

Since each $V(\vartheta_i)$ is in the MBL phase, we can find an effective Hamiltonian given by Eq.~\eqref{eq:app_mbl_phenom}.
In other words, we assume that there are MBL Hamiltonians $H^{(i)}_{\rm MBL}$ such that $V(\vartheta_i) = \exp[-i H^{(i)}_{\rm MBL} T]$ for some $T$.
Then, $U$ can be written as $U=\mathcal{T}[e^{-i\int_0^{pT} dt H(t)}]$ where $H(t)$ is defined by
\begin{align}
    H(t) = \begin{cases}
        H^{(1)}_{\rm MBL}& \text{ for } 0 \leq t < T \\
        H^{(2)}_{\rm MBL}& \text{ for } T \leq t < 2T \\
        &\cdots\\
        H^{(p)}_{\rm MBL}& \text{ for } (p-1)T \leq t < pT 
    \end{cases}.
\end{align}
As each $H^{(i)}_{\rm MBL}$ has a logarithmic lightcone, so does $H(t)$.

Therefore, under Conjecture~\ref{conj:log_lightncone_mbl}, we conclude that $U$ is in the MBL phase.
Conjecture~\ref{conj:log_lightncone_mbl} is widely accepted as a logarithmic lightcone explains many important properties of the MBL systems.
However, since there is no mathematically rigorous proof yet, a proof of Conjecture~\ref{conj:log_lightncone_mbl} might be studied in future work.

\subsection{The MBL phase of the hardware efficient ansatz with mutually commuting entangling gates}

So far, we have considered the MBL phase of the HEA composed of single-qubit RX, RZ, and CZ entangling gates.
In this subsection, we show that the HEA with commuting entangling gates can also have the MBL phase when full single-qubit rotation gates are allowed in the ansatz.

We consider a circuit $U(\pmb{\alpha}, \pmb{\beta}, \pmb{\gamma}) = \lprod_{i=1}^p V(\pmb{\alpha}_{i,:},\pmb{\beta}_{i,:},\pmb{\gamma}_{i,:})$, where each $V(\pmb{\alpha}_{i,:},\pmb{\beta}_{i,:},\pmb{\gamma}_{i,:})$ is given by
\begin{align}
    V(\pmb{\theta}_{i,:}) = \prod_{\braket{j,j'} \in E} W_{j,j'} \prod_{j=1}^N R(\alpha_{i,j},\beta_{i,j},\gamma_{i,j})
\end{align}
where $W_{j,j'}$ is an entangling gate and $R(\alpha_{i,j},\beta_{i,j},\gamma_{i,j})$ fully generate $\mathrm{SU}(2)$. Without loss of generality, we can assume that $\alpha_{i,j},\beta_{i,j},\gamma_{i,j}$ are the Euler angles.

We now assume that $W_{j,j'}$ are mutually commuting, i.e., $[W_{j,j'},W_{\tilde{j},\tilde{j}'}]=0$ for all $j,j',\tilde{j},\tilde{j}'$.
From Ref.~\cite{cubitt2016complexity}, we can find a local unitary operator $T \in \mathrm{SU}(2)$ such that $ W_{j,j'}  = T^{\otimes 2} D_{j,j'} (T^\dagger)^{\otimes 2}$ where $D_{j,j'}$ is a diagonal entangling gate.
We then obtain
\begin{align}
    V(\pmb{\theta}_{i,:}) = \prod_{\braket{j,j'}\in E}W_{j,j'} \prod_{j=1}^N R(\alpha_{i,j},\beta_{i,j},\gamma_{i,j}) = T^{\otimes N} \prod_{\braket{i,j}}D_{i,j}  (T^\dagger)^{\otimes N}  \prod_{j=1}^N R(\alpha_{i,j},\beta_{i,j},\gamma_{i,j}).
\end{align}
Further, by finding $\alpha_{i,j}',\beta_{i,j}',\gamma_{i,j}'$ such that $R(\alpha_{i,j}',\beta_{i,j}',\gamma_{i,j}') = T^\dagger R(\alpha_{i,j},\beta_{i,j},\gamma_{i,j}) T$, we can write
\begin{align}
    U = T^{\otimes N} \lprod_{i=1}^p V(\pmb{\alpha}_{i,:}',\pmb{\beta}_{i,:}',\pmb{\gamma}_{i,:}') (T^\dagger)^{\otimes N}.
\end{align}
Since the dynamic phase of the unitary operator does not depend on the local basis transformation, we can ignore the initial and final $T$ layers. 
Moreover, from the Euler decomposition, $R(\alpha_{i,j}',\beta_{i,j}',\gamma_{i,j}') = R_Z(\alpha_{i,j}')R_X(\beta_{i,j}')R_Z(\gamma_{i,j}')$.
By moving each $R_Z(\gamma_{i,j}')$ to the $i-1$-th block and combining it to $R_Z(\alpha_{i-1,j}')$, the phase of the circuit can be determined by the phase of the following unitary operator:
\begin{align}
    \widetilde{V} = \prod_{j=1}^N e^{-iZ_j \gamma_j/2} \prod_{\braket{j,j} \in E} D_{j j'} \prod_{j=1}^N e^{-i Z_j \alpha_j/2} \prod_{j=1}^N e^{-i X_j \beta_j/2},
\end{align}
where we renamed $\alpha',\beta',\gamma'$ to $\alpha,\beta,\gamma$ for convenience.

Since $D$ is a diagonal two-qubit gate, we can find $a,b,c,d \in \mathbb{R}$ such that 
\begin{align}
    D = \exp\Bigl[ i \bigl\{ a \Id \otimes \Id + b Z \otimes \Id + c \Id \otimes Z + d Z \otimes Z \bigr\} \Bigr].
\end{align}
By choosing $\gamma_j = 2 |E_{\rm in}(j)|b + |E_{\rm out}(j)|c$, where $E_{\rm in/out}(j)$ is the set of edges bounding to/from $i$, we can rewrite
\begin{align}
    \widetilde{V} = \prod_{\braket{j,j} \in E} e^{i d Z_j Z_{j'}} \prod_{j=1}^N e^{-i Z_j \alpha_j/2} \prod_{j=1}^N e^{-i X_j \beta_j/2}
\end{align}
up to a global phase.
This is nothing but the disordered Kicked Ising model studied in Refs.~\cite{ponte2015many,zhang2016floquet}.
Therefore, we can find parameters $\{\alpha_j\}$ and $\{\beta_j\}$ such that $\widetilde{V}$ is in the MBL phase, and so does $V(\pmb{\theta}_{i,:})$.

\section{Derivation of the gradient scaling in the MBL system} \label{app:grad_scaling_mbl}

In this section, we derive the scaling of gradients when the HEA is in the MBL system.
Our main tool is the phenomenological model introduced in Sec.~\ref{app:floquet-mbl}.

\subsection{Deriving the expressions of gradients for a single Pauli-Y and a multi-body observable}
In the main text, we considered the gradient component for $\theta_{p,1}$.
The expression involves the commutator $[X_1,\widetilde{Y_1}]$ where $\widetilde{Y_1} = V(\pmb{\theta}_{p,:})^\dagger Y_1 V(\pmb{\theta}_{p,:})$.
Then, $[X_1,\widetilde{Y_1}]$ is expanded as the sum of Pauli strings.
In this subsection, we derive this result.

As in the main text, we assign $\theta_{p,j} = \vartheta_p$ for $1 \leq j \leq N$.
From the definition of $\widetilde{Y_1}$, we have
\begin{align*}
    \widetilde{Y_1} &:= V(\pmb{\theta}_{p,:})^\dagger Y_1 V(\pmb{\theta}_{p,:}) = \prod_{j=1}^N e^{i X_j \vartheta_p/2} \prod_{j=1}^N e^{i Z_j \theta_{p,j+N}/2} \prod_{\braket{j, j'} \in E} \mathrm{CZ}_{j,j'} Y_1 \prod_{\braket{j, j'} \in E} \mathrm{CZ}_{j,j'} \prod_{j=1}^N e^{-i Z_j \theta_{p,j+N}/2} \prod_{j=1}^N e^{-iX_j \vartheta_p/2} \\
    &= e^{i X_1 \vartheta_p/2} e^{i Z_1 \theta_{p,N+1}/2} Y_1 e^{-i Z_1 \theta_{p,N+1}/2} e^{-i X_1 \vartheta_p/2} \prod_{j \in \mathcal{N}(i)} e^{i X_j \vartheta_p/2} Z_j e^{-i X_j \vartheta_p/2}, \numberthis \label{eq:tilde_y_def}
\end{align*}
where we used $\mathrm{CZ}_{j,j'} Y_1 \prod_{\braket{j, j'} \in E} \mathrm{CZ}_{j,j'} = Y_1 \prod_{j \in \mathcal{N} (i)} Z_j$, which is from the property of the CZ gate.

From the property of the rotation, we additionally have
\begin{align}
    e^{i X_1 \vartheta_p/2} e^{i Z_1 \theta_{p,N+1}/2} Y_1 e^{-i Z_1 \theta_{p,N+1}/2} e^{-i X_1 \vartheta_p/2} &=  e^{i X_1 \vartheta_p/2}\bigl[\cos(\theta_{p,N+1})Y_1 + \sin(\theta_{p,N+1})X_1\bigr] e^{-i X_1 \vartheta_p/2} \nonumber \\
    &= \cos(\theta_{p,N+1}) \bigl[ \cos(\vartheta_p) Y_1 - \sin(\vartheta_p) Z_1 \bigr] + \sin(\theta_{p,N+1})X_1
\end{align}
and
\begin{align}
    e^{i X_j \vartheta_p/2} Z_j e^{-i X_j \vartheta_p/2} = \cos(\vartheta_{p}) Z_j + \sin(\vartheta_p) Y_j.
\end{align}

Inserting these expressions to Eq.~\eqref{eq:tilde_y_def} yields
\begin{align}
    \widetilde{Y_1} = \Bigl\{ \cos(\theta_{p,N+1}) \bigl[ \cos(\vartheta_p) Y_1 - \sin(\vartheta_p) Z_1 \bigr] + \sin(\theta_{p,N+1})X_1 \Bigr\} \prod_{j \in \mathcal{N}(1)} \bigl[\cos(\vartheta_{p}) Z_j + \sin(\vartheta_p) Y_j \bigr]. \label{eq:tilde_y_pstr}
\end{align}

As a consequence, we obtain
\begin{align}
    [X_1, \widetilde{Y_1}] = 2i \cos(\theta_{p,N+1}) \bigl[  \cos(\vartheta_{p}) Z_1 +  \sin(\vartheta_{p}) Y_1  \bigr]  \prod_{j \in \mathcal{N}(1)} \bigl[\cos(\vartheta_{p}) Z_j + \sin(\vartheta_p) Y_j \bigr].
\end{align}

Now, let us consider $O=Y_1 \prod_{j=1}^N Z_j$. Following a similar step, we obtain
\begin{align}
    \widetilde{O} &= V(\pmb{\theta}_{p,:})^\dagger \Bigl[ Y_1 \prod_{j=1}^N Z_j \Bigr] V(\pmb{\theta}_{p,:}) \nonumber \\ 
    &= e^{i X_1 \vartheta_p/2} e^{i Z_1 \theta_{p,N+1}/2} Y_1 e^{-i Z_1 \theta_{p,N+1}/2} e^{-i X_1 \vartheta_p/2} \prod_{j \in S} e^{i X_j \vartheta_p/2} Z_j e^{-i X_j \vartheta_p/2}, \numberthis \label{eq:tilde_yprodz_def}
\end{align}
where $S=\{2,\cdots\,N\} \setminus \mathcal{N}(1)$. Therefore, the commutator becomes
\begin{align}
    [X_1, \widetilde{O}] = 2i \cos(\theta_{p,N+1}) \bigl[  \cos(\vartheta_{p}) Z_1 +  \sin(\vartheta_{p}) Y_1  \bigr]  \prod_{j \in S} \bigl[\cos(\vartheta_{p}) Z_j + \sin(\vartheta_p) Y_j \bigr]. \label{eq:comm_x_yzz}
\end{align}

\begin{figure}
    \centering
    \includegraphics[width=0.95\linewidth]{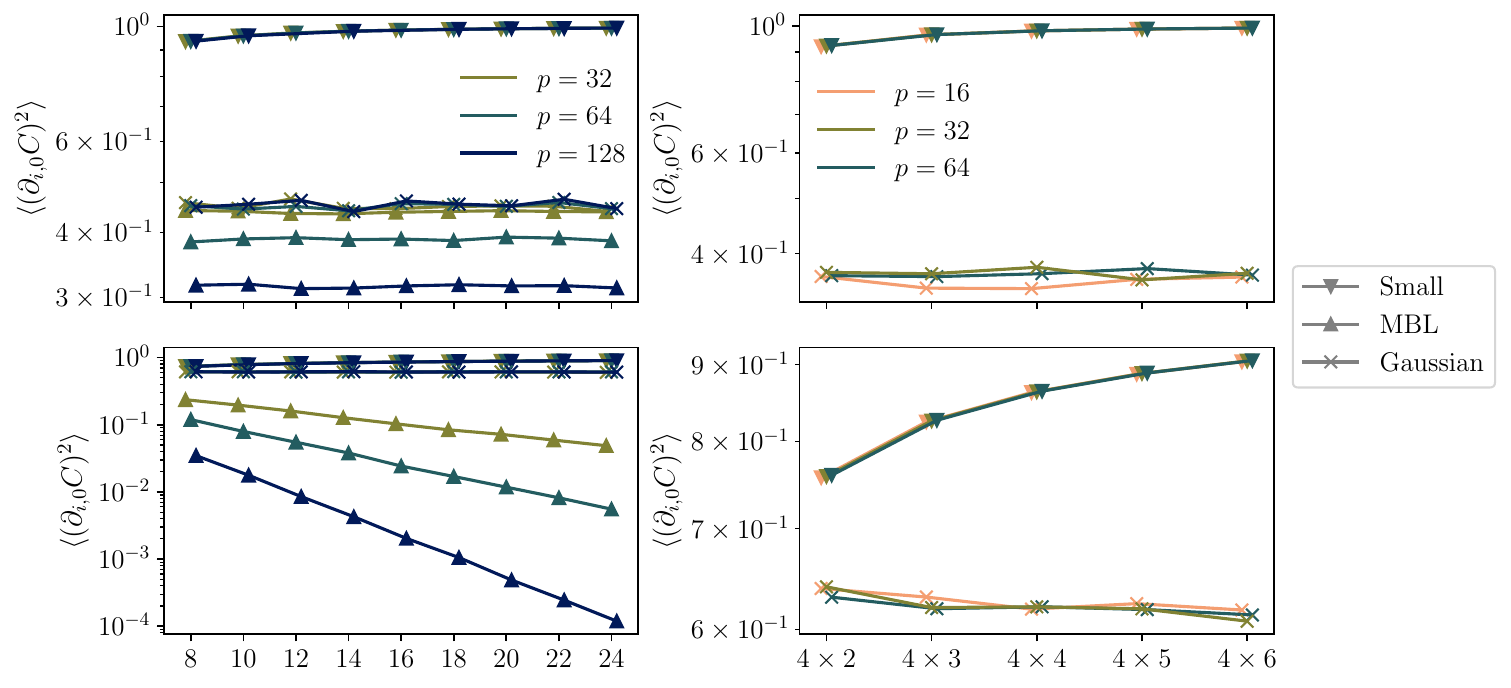}
    \caption{Scaling of gradients for the 1D (left column) and 2D HEAs (right colume) with observables $O=Y_1$ (first row) and $O=Y_1\prod_{j=2}^N Z_j$ (second row). The number of blocks $p\in [32,64,128]$, $p \in [16, 32,64]$ are used for the 1D and 2D HEA, respectively.
    The weight of the observable is given by $S=1$ for $O=Y_1$ and $S=N$ for $O=Y_1\prod_{j=2}^N Z_j$.
    }
    \label{fig:compare-grads-gaussian}
\end{figure}

\subsection{Long-time limit of gradients}
In this subsection, we derive the long-time limit of the gradient component for $\theta_{p,1}$, which is given by
\begin{align}
    \partial_{p,1}C = \frac{i}{2}\braket{0^N|\prod_{i=1}^{N-1} V(\pmb{\theta}_{i,:})^\dagger [X_1, \widetilde{Y_1}] \lprod_{i=1}^{N-1} V(\pmb{\theta}_{i,:}) |0^N},
\end{align}
where $\prod_{i=1}^k U_i = U_1 \cdots U_k$ and $\lprod_{i=1}^k U_i = U_k\cdots U_1$.

Now let us assume that there is a phenomenological model for $U=V(\pmb{\theta}_{N-1,:})\cdots V(\pmb{\theta}_{1,:})$, i.e., there exists a Hamiltonian, $H_{\rm MBL}$, in the form of Eq.~\eqref{eq:app_mbl_phenom} such that $U=e^{-iH_{\rm MBL} (p-1)T}$.
From the definition of the local integrals of motion, $\tau^{i}_z$ has a finite overlap with $Z_i$.
Namely, there is a constant that lower bounds $A_i=\Tr[\tau^i_z Z_i]/2^N$.
In addition, from Eq.~\ref{eq:tilde_y_pstr}, we obtain
\begin{align}
    \mathbb{E}_{\pmb{\theta}}[V(\pmb{\theta})^\dagger Y_iV(\pmb{\theta})] = 0.
\end{align}
Therefore, it is natural to assume that $Y_j$ overlaps little with any products of $\{\tau^i_z\}$.

Let $S$ be a subset of $[N]$ and $Z_S = \prod_{i \in S} Z_i$. Then, the multi-point correlation function after time $t$ is given by
\begin{align}
    \braket{Z_S(t)} = \braket{0^N|e^{iH_{\rm MBL}t} Z_S e^{-iH_{\rm MBL}t}| 0^N}.
\end{align}
Inserting $Z_i = A_i \tau^i_z + O(e^{-\xi^{-1}})$, we obtain
\begin{align}
    \braket{Z_S(t)}&= \braket{0^N|e^{iH_{\rm MBL}t} \prod_{i \in S} Z_i e^{-iH_{\rm MBL}t}| 0^N} \\
    &\xrightarrow{t \gg 1} \braket{0^N|\prod_{i \in S} A_i \tau^i_z | 0^N} \\
    &\approx \braket{0^N|\prod_{i \in S} A_i^2 Z_i | 0^N} =\prod_{i \in S} A_i^2
\end{align}
in a deep localized phase ($\xi \ll 1$). When averaged over the disorder realization, $A^2:=\mathbb{E}_{\pmb{\theta}}[A_i^2]$ becomes independent of $i$.

Similarly, we obtain that the multi-point correlation functions containing $Y_j$ vanish when averaged over the disorders, i.e., correlation functions such as $\braket{Z_1Y_2(t)}$ become $0$ when averaged over the disorders.
Thus, it is also natural to assume that those correlation functions are small for sufficiently large $t$.
In fact, it is known that such a quantity decays polynomial with $t$ until it saturates for each disorder realization~\cite{serbyn2014quantum}, where the saturated values decay with $N$.

In summary, we obtain
\begin{align*}
    \partial_{p,1}C &= -\Bigl\langle 0^N \Bigl\vert e^{iH_{\rm MBL}(p-1)T}\cos(\theta_{p,N+1}) \bigl[  \cos(\vartheta_{p}) Z_1 +  \sin(\vartheta_{p}) Y_1  \bigr]  \\
    & \qquad \times \prod_{j \in \mathcal{N}(1)} \bigl[\cos(\vartheta_{p}) Z_j + \sin(\vartheta_p) Y_j \bigr]e^{-iH_{\rm MBL}(p-1)T} \Bigr\vert 0^N \Bigr\rangle \\
    &\approx -\cos(\theta_{p,N+1}) \cos(\vartheta_{p})^{1+|\mathcal{N}(1)|} \Bigl\langle \prod_{i \in \{1\}\cup \mathcal{N}(1)}Z_i (t)\Bigr\rangle \\
    &\xrightarrow{t \gg 1} -\cos(\theta_{p,N+1}) [\cos(\vartheta_{p})A^2]^{1+|\mathcal{N}(1)|}, \numberthis
\end{align*}
where $\mathcal{N}(1)$ is the set of all sites which is connected to $1$.
The final expression is what we used in the main text.

Following the same arguments, we can compute the gradient for $O=Y_1\prod_{j=2}^N Z_j$. From Eq.~\eqref{eq:comm_x_yzz}, we obtain
\begin{align}
    \partial_{p,1}C &\xrightarrow{t \gg 1} -\cos(\theta_{p,N+1}) [\cos(\vartheta_{p})A^2]^{N-|\mathcal{N}(1)|}.
\end{align}

\section{Numerical results for the 2D hardware efficient ansatz and Gaussian initialization} \label{app:additional_numerical_results}
In this section, we numerically compare our initialization methods with the Gaussian method introduced in Ref.~\cite{zhang2022gaussian}.
In addition to the parameter setups \textbf{Small} and \textbf{MBL}, used in the main text, we add the following. \textbf{Gaussian}: All parameters are sampled from the Gaussian distribution $\mathcal{N}(0, [(Sp)^{-1/2}]^2)$ where $S$ is the weight of the observable.

We plot our results in Fig.~\ref{fig:compare-grads-gaussian}.
The same data as in the main text is used for the 1D HEA with the Small and MBL schemes.
The MBL results are only shown for the 1D HEA as the MBL transition point is unknown for the 2D HEA.
Interestingly, both for 1D and 2D HEAs, we observe that the Gaussian initialization also gives $\Theta(1)$ gradient magnitudes. 
This is not explained by our theorem nor the main result in Ref.~\cite{zhang2022gaussian}.

In general, we expect that there exists a transition point $\tau_c(N)$ such that the averaged gradient magnitudes is $\Theta(1)$ for all $\sum_{ij}|\theta_{ij}| \leq \tau_c(N)$.
Theorem~\ref{thm:hea_small_param_large_grad_app} tells that $\tau_c(N) = \Omega(1/N)$, and our numerical results here suggest that $\tau_c(N) =\Omega(N^{-1/2})$.

A more in-depth numerical study is necessary to locate the exact transition point, including the data collapse analysis.
As such a study is beyond the scope of this work, this question will be addressed in further study.

\section{Machine learning application} \label{app:machine_learning}

In the main text, we have studied the performance of the HEA for solving quantum many-body spin models.
In this section, we solve a classification problem, a typical supervised learning task~\cite{schuld2018supervised}, using the HEA.
We discuss how the performance of the model can be improved using our initialization schemes.

\subsection{Dataset}
We use the MNIST PCA dataset, which is generated by processing the original MNIST dataset using the principal component analysis (PCA).
From the original dataset, we extract pixel data for digits 3 and 5. The data is processed using the PCA with the output dimension of $d=20$.
Training and test sets are obtained by extracting 250 random output vectors.
We also assign each label $\pm 1$ if the digit was $3$ and $5$, respectively.

\subsection{Encoding}
We encode $d$ dimensional data to a circuit using the angle encoding applied to qubits $\{N-d+1, \cdots, N\}$, followed by an entangling layer.
Precisely, for each data point $\pmb{\phi} = \{\phi_{i}\}_{i=1}^{d}$, we first normalized them
\begin{align}
    \widetilde{\phi}_i = \frac{\phi_i}{\Vert \pmb{\phi} \Vert_2},
\end{align}
where $\Vert \pmb{\phi} \Vert_2^2 = \sum_{i=1}^d \phi_i^2$ is the L2 norm.

We then encode the normalized data point using the following encoding gate:
\begin{align}
    E(\pmb{\phi}) = \prod_{i=1}^{N-1} \mathrm{CZ}_{i, i+1} \prod_{i=N-d}^{N} e^{-i\widetilde{\phi}_i X_{i}/2}.
\end{align}

\begin{figure}
    \centering
    \includegraphics[width=0.8\linewidth]{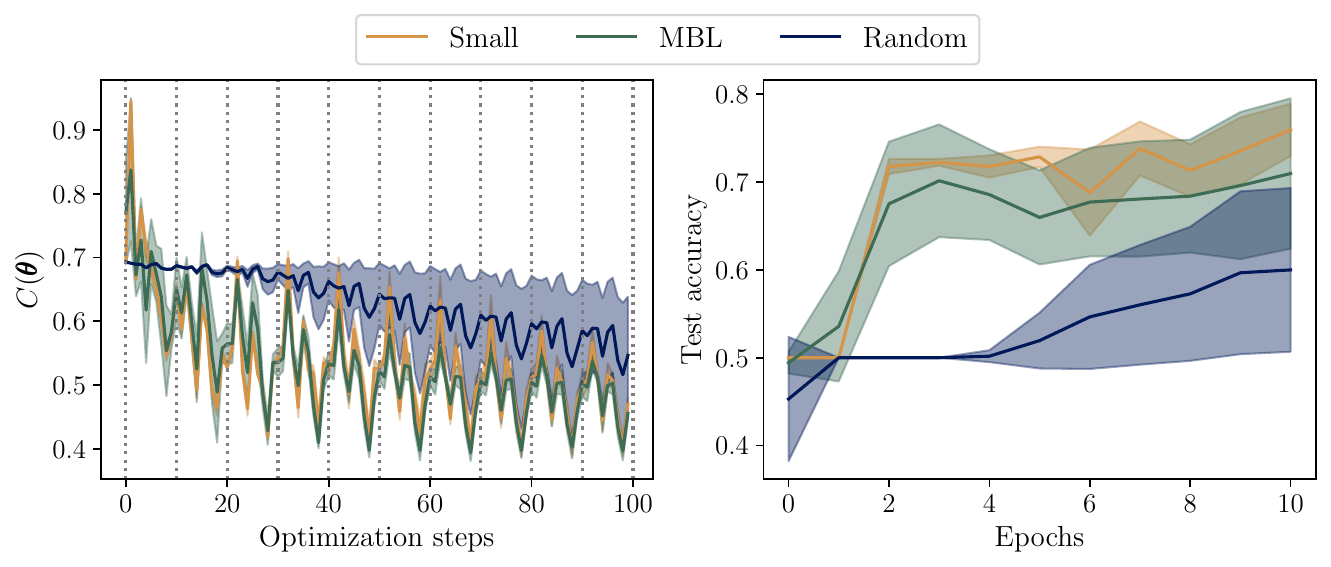}
    \caption{For the MNIST PCA dataset with $d=20$, we train a quantum machine learning model based on the HEA. The parameters of the circuits are initialized following three distributions: Small, MBL, and Random.
    The loss function for each optimization step (Left) and the test accuracy as a function of epochs (Right) are shown. The size of the training set is $250$, and the minibatch size $B=25$ is used, i.e., each epoch has 10 optimization steps. 
    The dotted lines on the left figure indicate each epoch.
    After each epoch, we compute the test accuracy using a test set size of $250$. The shaded regions indicate $[m-\sigma, m+\sigma]$ where $m$ and $\sigma$ are the mean value and the standard deviation obtained from $16$ independent instances, respectively.}
    \label{fig:qml_mnist_pca}
\end{figure}

\subsection{Cost function and gradient}
We use the HEA for the last part of the circuit and $\ket{0}^{\otimes N}$ as the initial state. Namely, the output state of the circuit is given by
\begin{align}
    \ket{\psi(\pmb{\theta}, \pmb{\phi})} = \lprod_{i=1}^d V(\pmb{\theta}_{i,:})E(\pmb{\phi}) \ket{0}^{\otimes N},
\end{align}
where
\begin{align}
    V(\pmb{\theta}_{i,:}) = \prod_{j=1}^{N-1} CZ_{j,j+1}\prod_{j=1}^N e^{-i Z_j \theta_{i,j+N}/2} \prod_{j=1}^N e^{-i X_j \theta_{i,j}/2}.
\end{align}

As we solve a binary classification problem, we interpret the expectation value as the probability as follows:
\begin{align}
    p(\pm 1;\pmb{\theta}, \pmb{\phi}) = \Bigl\langle \psi(\pmb{\theta}, \pmb{\phi}) \Bigl \vert \frac{1\pm Y_1}{2} \Bigr \vert \psi(\pmb{\theta}, \pmb{\phi}) \Bigr \rangle.
\end{align}

For each training iteration, we choose a minibatch of size $B$ from the training set. We denote $\{\pmb{\phi}_{k}\}_{k=1}^B$ by the data vectors and the corresponding labels $\{y_k\}_{k=1}^B$ Then, we use the binary cross entropy loss as the cost function, defined as
\begin{align}
    C(\pmb{\theta}) = - \frac{1}{B} \sum_{k=1}^B \Bigl[ p(y_k = +1) \log p(+ 1;\pmb{\theta}, \pmb{\phi}_k) + p(y_k = -1) \log p(- 1;\pmb{\theta}, \pmb{\phi}_k) \Bigr],
\end{align}
where $p(y_k = +1) = (1+y_k)/2$, which is $1$ if $y_k=1$, and $0$ if $y_k = -1$. The probability $p(y_k=-1)$ is defined similarly.

When the dimension of the number of qubits is larger than the input data, i.e., $N > d$, one can check that the encoding gate does not affect the gradient for the RX gates acting on the first qubit when $\pmb{\theta}=0$. For example, 
\begin{align}
    \partial_{1,1} p(\pm 1;\pmb{\theta}, \pmb{\phi})|_{\pmb{\theta}=0} = \frac{\partial_{1,1} \langle \psi(\pmb{\theta}, \pmb{\phi})| Y_1 |\psi(\pmb{\theta}, \pmb{\phi}) \rangle }{2} \Bigr|_{\pmb{\theta}=0} = \frac{i\braket{0^N|E(\pmb{\phi})^\dagger[X_1, Y_1]E(\pmb{\phi})|0^N}}{2} = -\frac{1}{2},
\end{align}
where the last inequality follows from the fact that $E(\pmb{\phi})$ only acts on qubits $\{N-d+1,\cdots,N\}$.
In addition, by considering $\ket{\psi_0} = E(\pmb{\phi})\ket{0^N}$ as the initial state, Theorem~\ref{thm:hea_small_param_large_grad_app} can be applied to this setup.

When the circuit parameters are completely random, the circuit forms a 1-design, which implies $p(+ 1;\pmb{\theta}, \pmb{\phi}) \approx 1/2$.
Thus, the initial gradient of $C(\pmb{\theta})$ is exponentially small when parameters are initialized following Random, while it has non-zero components for Small and MBL.

\subsection{Numerical results}
We train our quantum machine learning model based on the HEA with $p=128$ by optimizing $C(\pmb{\theta})$ using Adam. 
We use a minibatch size of $B=25$ (thus, total $10$ iterations for each epoch) and a learning rate of $\eta = 0.01$. Default values $\beta_1=0.9$, $\beta_2=0.999$, and $\epsilon=10^{-8}$ are used for other hyperparameters.

We plot the results from our numerical simulation in Fig.~\ref{fig:qml_mnist_pca}. Exactly computed gradients are used to optimize our loss function $C(\pmb{\theta})$.
We observe that the model performs the best with the Small initialization, marginally followed by the MBL. The result from Random is the worst, showing that the loss function decays very slowly, and the test accuracy does not increase over the first few training epochs.
These results are consistent with that for solving many-body Hamiltonians, which we studied in the main text.

\end{document}